\documentclass[a4paper,10pt]{article}

\usepackage{graphicx}
\usepackage{amssymb}
\usepackage{latexsym}

\newtheorem{thm}{Theorem}
\newtheorem{lem}[thm]{Lemma}
\newtheorem{corl}[thm]{Corollary}

\newenvironment{proof}{{\bf Proof. } } 

\def\qed{\hfill $\Box$}

\title{Search by a Metamorphic Robotic System \\ 
in a Finite 2D Square Grid} 

\author{Keisuke Doi\thanks{Graduate School of Information Science and Electrical Engineering, Kyushu University, Japan. } 
\and 
Yukiko Yamauchi\thanks{Corresponding author. Faculty of Information Science and Electrical Engineering, Kyushu University, Japan.  e-mail: \texttt{yamauchi@inf.kyushu-u.ac.jp}} 
\and 
Shuji Kijima \thanks{Faculty of Information Science and Electrical Engineering, Kyushu University, Japan. }
\and 
Masafumi Yamashita\thanks{Faculty of Information Science and Electrical Engineering, Kyushu University, Japan. }}

\begin{document}
\date{}
\maketitle

\begin{abstract}
 We consider search in an unknown finite 2D square grid by
 a metamorphic robotic system consisting of anonymous 
 memory-less modules.
 Each module autonomously moves with executing 
 a common distributed algorithm and 
 the modules collectively form a robotic system 
 by keeping connectivity. 
 The number of shapes of the metamorphic robotic 
 system grows as the number of modules increases, and 
 a shape of the system serves as its memory and show 
 its functionality. 
 We present the minimum number of modules for search in 
 a finite 2D square grid. 
 We demonstrate that if the modules agree on the directions, 
 i.e., they are equipped with the global compass, 
 three modules are necessary and sufficient for
 search from an arbitrary initial shape, 
 otherwise five modules are necessary and sufficient 
 for search from limited initial shapes. 


\end{abstract}

\section{Introduction}

As autonomous mobile robots, drones, and vehicles become widely spread 
and nanoscale manufacturing enables synthesis of 
target DNA, RNA, and polymers, 
computational power of a collection of mobile computing entities
with very weak capabilities attracts much attention 
in distributed computing theory. 
The investigation also reveals the minimum capabilities for the 
mobile computing entities to accomplish a given task and 
provides design guidelines for hardware robots. 
These results are also expected to give a clue to understanding
complex behavior of natural systems. 
Each mobile computing entity is called robots, agents, or particles, and 
it is often assumed that it is \emph{anonymous} and 
\emph{oblivious} (memory-less), and
does not have any communication capability
or any access to the global coordinate system. 
Many distributed computing models have 
proposed for modular robots~\cite{DSY04a}, 
mobile robots~\cite{SY99}, 
mobile ad-hoc networks~\cite{AADFP06}, 
biological systems~\cite{DGRSS16}, and 
chemical reactions~\cite{MS17}. 
There are a variety of indicator tasks such as 
gathering~\cite{CFPS12,SY99,YS10}, leader election~\cite{DGRSS16,DFSVY17}, 
decomposition~\cite{LYKY18}, 
shape formation~\cite{DGRSS16,DFSVY17,DSY04a,FYOKY15,MSS19,MS17,SY99,YS10,YUKY17,YY13,YY14}, and 
computing functions~\cite{AADFP06,MS17}. 

Suzuki and Yamashita investigated the \emph{pattern formation problem} 
by anonymous \emph{mobile robots} moving in the continuous 2D space~\cite{SY99}. 
Each robot observes the positions of other robots, 
computes its next position with a common deterministic algorithm, 
and moves to the next position. 
The pattern formation problem requires that the mobile robots form 
a specified target pattern (shape) from a given initial configuration. 
They pointed out that the 
pattern formation problem is essentially related to
the \emph{agreement problem}. 
For example, the simplest agreement problem is the point formation problem, 
also called the gathering problem, 
that requires the robot to agree on a single point.
Additionally, once the robots agree on a common coordinate system,
they can form any target pattern. 
However, the anonymous robots cannot break the symmetry of their initial 
configuration, and the set of formable patterns is characterized 
by the rotational symmetry of the initial configuration. 

Derakhshandeh et al. presented a shape formation algorithm
for the \emph{amoebot model} in the 2D triangular grid~\cite{DGRSS16}. 
The amoebot model consists of anonymous programmable particles 
each of which is equipped with constant-size memory and communication
capability with neighboring particles. 
Each vertex of a triangular grid is occupied by at most one particle 
at each time step, and 
each particle moves by repeating an extension and a contraction.
Their algorithm is based on a randomized \emph{leader election}, 
that allows formation of an arbitrary shape consisting of triangles. 

Dumitrescu et al. considered shape formation 
in the \emph{metamorphic robotic system model} 
that consists of oblivious modules moving 
in the 2D square grid~\cite{DSY04a,DSY04b}.
Each module autonomously perform two types of local movements,
called a rotation and a sliding; however they are required to 
maintain connectivity. 
They presented a canonical shape,
to which any convex shape can be transformed. 
The reversibility of movements guarantees
transformation between any pair of convex shapes via the canonical shape. 

The goal of these existing studies is the structure of shapes 
of a distributed system consisting of mobile computing entities. 
Reachability among shapes decomposes the system's configuration
space into subspaces, 
that indicate the degree of agreement and coordination, 
in other words, computational power of the system. 
When we take a closer look at existing shape formation
algorithms, we find that intermediate shapes are used 
to guarantee the progress of distributed coordination. 
In other words, a shape of the system serves as its global memory 
while each computing entity is memory-less or
equipped with constant-size memory. 

In this paper, we investigate the functionality of 
shapes of a distributed system. 
We focus on the \emph{search problem} 
by the metamorphic robotic system. 
The problem requires the metamorphic robotic system to find a target 
put in one cell of a given field, which is a finite rectangular 
subspace of the 2D square grid. 
Each module is not given any apriori knowledge of the field or 
the target cell and 
the metamorphic robotic system may start from an arbitrary initial 
configuration. 
To find the target, the metamorphic robotic system needs to remember, 
for example, its starting point, its current moving direction, and the 
progress of search. 
Clearly, as the number of modules increases,
the number of possible shapes increases and 
the system can store more information. 
We investigate the minimum number of modules for search in an unknown 
2D square grid and the effect of the \emph{global compass} that 
allows the modules to agree on north, south, east, and west. 

\noindent{\bf Our results.~} 
Although shape formation~\cite{DSY04a,MSS19}
and locomotion~\cite{CYKY14,DSY04b} 
by the metamorphic robotic system have been discussed, 
to the best of our knowledge, this is the first time 
the search problem is discussed for a metamorphic robotic system. 
We first demonstrate when the modules are equipped with 
the global compass,
three modules are necessary and sufficient for search 
from an arbitrary initial configuration. 
We show impossibility for less than three modules and 
present a distributed search algorithm for three modules. 
The proposed algorithm is based on the \emph{exploration} of a given field, 
i.e., the metamorphic robotic system visits each cell of a given field. 
The advantage of the proposed algorithm is 
that the metamorphic robotic system can start from 
an arbitrary initial configuration, i.e., 
the algorithm guarantees \emph{self-stabilization}. 
A self-stabilizing distributed algorithm guarantees 
that the system autonomously converges to a legitimate execution 
from an arbitrary initial system configuration~\cite{D74}. 
Self-stabilization promises robustness against a finite number of 
transient faults 
and autonomous adaptability to intermittent environment changes. 
Note that search and exploration 
by a single metamorphic robotic system
is different from 
search by ants~\cite{ELSUW15},
mobile agents~\cite{DFKNS07}, or 
mobile robots~\cite{BMPT11,DLPRT12,FIPS13} 
because the metamorphic robotic system cannot separate 
into several small fragments.

We next demonstrate that when the modules are 
not equipped with the global compass,
five modules are necessary and sufficient; 
however there are initial shapes from which
the metamorphic robotic system cannot perform search 
due to symmetric deadlocks. 
We present a distributed algorithm for five modules from 
any initial configurations without deadlocks. 
We also demonstrate that less than five modules 
cannot perform search in a given field 
even when we exclude initial deadlocks. 

We finally consider self-stabilizing search 
when the modules are not equipped with the global compass. 
As a primitive, we consider \emph{locomotion} 
that requires a metamorphic robotic system to keep on moving 
to one direction. 
The metamorphic robotic system must break their initial symmetry 
so that they agree on a moving direction. 
We demonstrate that seven modules can perform locomotion 
from any initial configuration. 
We thus conclude seven modules are necessary and sufficient 
for self-stabilizing search. 

\noindent{\bf Related works.~} 
Most existing papers on the computational power of 
distributed system consisting of mobile computing entities 
focus on \emph{shape formation}. 
One of the most active areas is pattern formation by 
autonomous mobile robots and the effect of various capabilities such as 
obliviousness~\cite{SY99},
asynchrony~\cite{FYOKY15,YS10}, 
limited visibility~\cite{YY13},
and randomness~\cite{YY14} 
has been discussed. 
These papers showed that 
rotational symmetry of an initial configuration determines formable shapes, 
i.e., obliviousness and asynchrony generally have no effect. 
Randomness allows probabilistic symmetry breaking and
realizes universal pattern formation. 
These results consider mobile robots in the 2D space, and 
Yamauchi et al. extended them to mobile robots in the 3D space, 
where rotational symmetry consists of 
five types of rotation groups~\cite{YUKY17}.

Shape formation in the amoebot model is investigated 
for shapes consisting of triangles~\cite{DGRSS16} 
and arbitrary shapes~\cite{DFSVY17}. 
Di Luna et al. considered the limit of deterministic
leader election and characterized formable shapes
by the symmetry of an initial configuration~\cite{DFSVY17}. 
Derakhshandeh et al. presented a universal 
shape formation algorithm 
based on a randomized leader election~\cite{DGRSS16}.  

Shape formation for the metamorphic robotic system 
is investigated
in a distributed setting and in a centralized setting. 
Dumitrescu et al. considered distributed transformability of an
initial (horizontally) convex shape to a line (also called a chain)
shape with the global compass~\cite{DSY04a}. 
Dumitrescu et al. considered \emph{locomotion} of the metamorphic
robotic system
and presented shapes that realizes fastest locomotion~\cite{DSY04b}. 
While these two papers assume unlimited visibility, 
Chen et al. considered locomotion with limited visibility~\cite{CYKY14}. 
Michail et al. considered decidability of centralized transformation~\cite{MSS19}. 
When the movement is limited to rotations,
there are pairs of shapes that are not transformable to each other. 
They showed that deciding whether an initial connected shape can be transformed 
to a target connected shape only by rotations is in $\mathrm{P}$ if 
disconnection during transformation is allowed, 
otherwise the decision problem is in $\mathrm{PSPACE}$. 
Then they showed combination of rotations and slidings allows 
transformation of any initial connected shape to a target connected shape. 

Michail and Spirakis proposed the \emph{network constructor model} 
that consists of finite-state agents under passive movement~\cite{MS17}.
The interaction model is based on the
\emph{population protocol model}~\cite{AADFP06},
while the agents can construct an edge when they interact. 
They considered distributed transformation to a spanning line 
and showed that $n$ agents can compute any symmetric predicate
in $\mathrm{SPACE}(n^2)$. 

All these papers consider reachability and classification of shapes. 
Little is investigated for the functionality of shapes.
Das et al. investigated the formation of a sequence of patterns,
that also serves as finite memory formed by oblivious mobile
robots~\cite{DFSY15}. 
Simulating a Turing machine by a line shape of computing entities 
has been separately discussed for
the metamorphic robotic system model~\cite{DSY04a}, 
the network constructor model~\cite{MS17}, and
the amoebot model~\cite{DFSVY17}.
Di Luna et al. showed a constant number of oblivious mobile 
robots can
simulate a robot with memory in a continuous space of an 
arbitrary dimension~\cite{DFSV18}. 
In this paper, we focus on the fact that
geometric configuration of the metamorphic robotic system 
functions as memory and processor, 
and
we investigate how a small number of oblivious modules 
collectively perform search in an unknown rectangular field.

\noindent{\bf Organization.~} 
We introduce the metamorphic robotic 
system model and the search problem in Section~\ref{sec:prel}. 
In Section~\ref{sec:global} 
we consider modules equipped with the global compass 
and present the minimum number of modules for 
search from an arbitrary initial configuration. 
In Section~\ref{sec:local} 
we consider modules not equipped with the global compass. 
We first illustrate initial deadlock states (shapes) 
caused by symmetric positions of modules. 
We then present the minimum number of modules 
for search from an initial configuration without deadlock. 
Finally, we discuss the minimum number of modules 
for symmetry breaking and give 
the minimum number of modules for locomotion 
from an arbitrary initial configuration. 
We conclude this paper with Section~\ref{sec:concl}.

\section{Preliminary}
\label{sec:prel}

We consider the rectangular metamorphic robotic system introduced in
\cite{CYKY14,DSY04b,DSY04a,MSS19}. 
Consider a two dimensional (2D) square grid where each square cell $c_{i,j}$ 
is labeled by the underlying $x$-$y$ coordinate system. 
We consider a finite subspace of width $w$ and height $h$ and
call it the \emph{field}. 
Without loss of generality, we assume that $c_{0,0}$ is the
southwesternmost cell and $c_{w-1, h-1}$ is the northeasternmost cell 
(Figure~\ref{fig:subgrid}).  
Each cell $c_{i,j}$ has eight \emph{adjacent} cells; 
(E)ast $c_{i+1, j}$,
(N)orth(E)ast $c_{i+1, j+1}$,
(N)orth $c_{i, j+1}$,
(N)orth(W)est $c_{i-1, j+1}$,
(W)est $c_{i-1, j}$,
(S)outh(W)est $c_{i-1,j-1}$,
(S)outh $c_{i, j-1}$, and 
(S)outh(E)ast $c_{i+1, j-1}$. 
The four cells N, S, E, and W
are said to be \emph{side-adjacent} to $c_{i,j}$. 
An infinite sequence of cells with the same $x$ coordinate is called
a \emph{column} and
an infinite sequence of cells with the same $y$ coordinate is called
a \emph{row}. 
The field is surrounded by walls,
the $(-1)$st column (the west wall), 
the $w$th column (the east wall),
the $(-1)$st row (the south wall), and
the $h$th row (the north wall). 
These cell labels are used just for description and there is no 
way to distinguish cells.

A metamorphic robotic system $R$ consists of $n$ anonymous modules,
each of which occupies a distinct cell in the square grid at
discrete time steps $t=0, 1, 2, \ldots$.
The \emph{configuration} $C_t$ of $R$ at time $t$ is the set of cells
occupied by the modules at time $t$.
An \emph{execution} is an evolution of configurations
$C_0, C_1, C_2, \ldots$. 
The evolution is generated by movements of modules.
Let $M_t$ be the set of modules that move at time $t$. 
We call the modules in $B_t = C_t \setminus M_t$ a \emph{backbone},
that does not move at time $t$. 
There are two types of movements, a \emph{rotation} and a \emph{sliding},
guided by backbone modules (Figure~\ref{fig:movement}). 
A rotation of module $m$ side-adjacent to a backbone module $b$
is a rotation around $b$ by $\pi/2$ either clockwise or
counter-clockwise.
A $1$-sliding of module $m$ is a sliding to a side-adjacent cell. 
In this case, there must be two backbone modules; 
one is $b_1$ that is side-adjacent to $m$ and
the other is $b_2$ that is side-adjacent to $b_1$ and
the destination cell of $m$. 
A $k$-sliding ($k \geq 2,3, \ldots$) is defined in the same way 
and it requires $(k+1)$ backbone modules along the track.\footnote{
The metamorphic robotic system model in
\cite{CYKY14,DSY04b,DSY04a,MSS19} allows rotations and $1$-slidings.
We extend the original model by allowing $k$-slidings for
$k = 2, 3, \ldots$. 
We assume $k$ is constant with respect to the size of the field.}
In a rotation and sliding, 
the cells that $m$ passes must not contain any module. 
The modules cannot enter or pass through the cells of the walls.

\begin{figure}[t]
\centering
 \includegraphics[height=5cm]{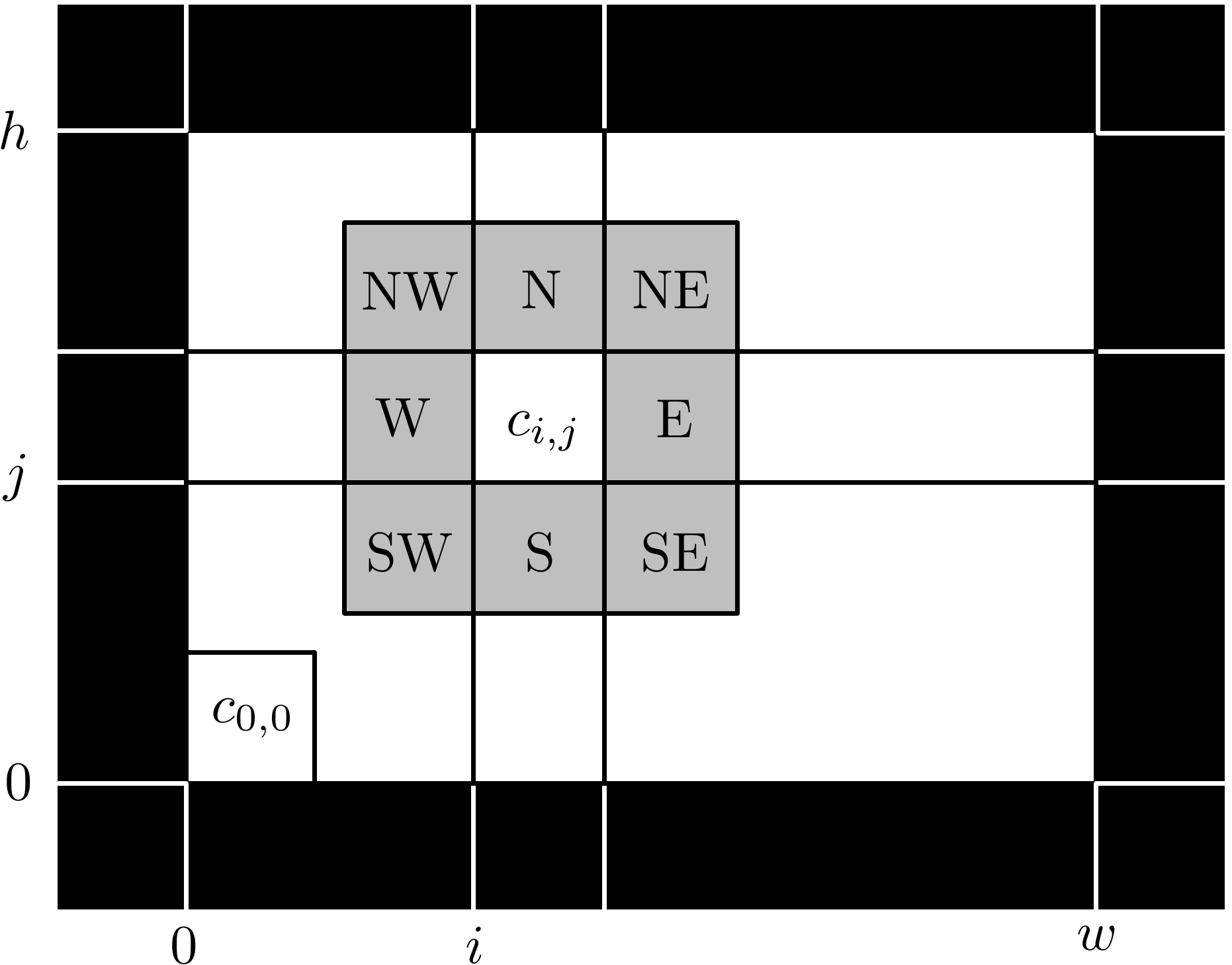}
 \caption{Field and walls.} 
 \label{fig:subgrid}
\end{figure}

\begin{figure}[t]
\centering
 \includegraphics[height=2cm]{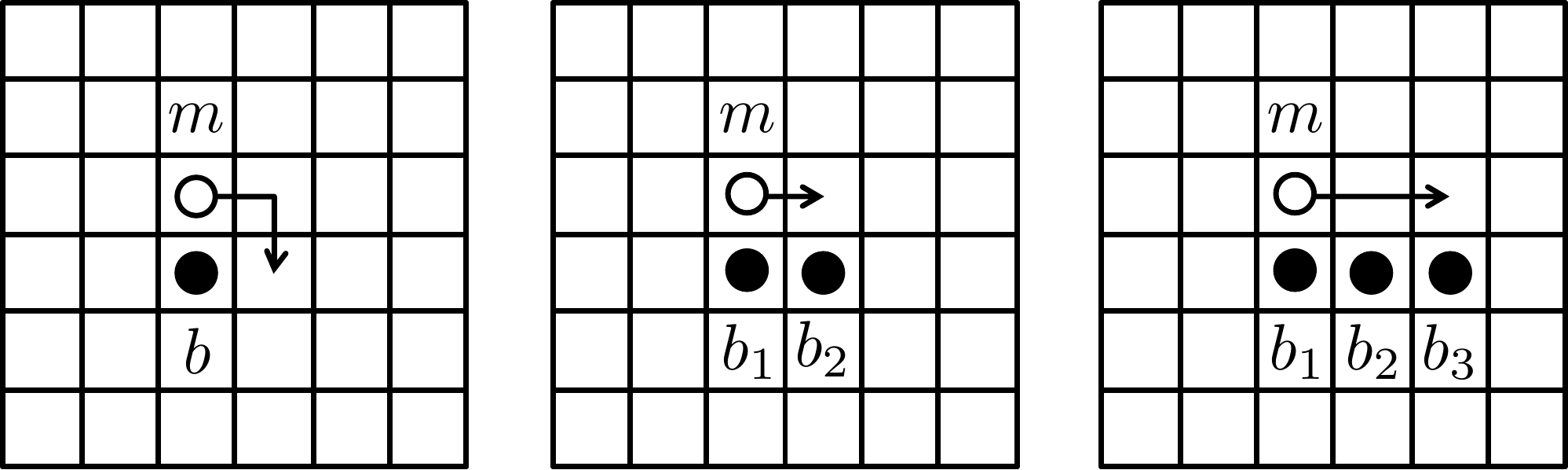}
 \caption{Rotation, $1$-sliding, and $2$-sliding}
 \label{fig:movement}
 \end{figure}

The connectivity of configuration $C_t$ is represented by
a connectivity graph $G_t =(C_t, E_t)$. 
The edge set $E_t$ contains an edge $(c,c')$ for $c, c' \in C_t$ 
if and only if cells $c$ and $c'$ are side-adjacent. 
When $G_t$ is connected, we say $C_t$ is \emph{connected}. 
Any execution $C_0, C_1, C_2, \ldots$ must satisfy the following
three conditions:
\begin{enumerate}
 \item Connectivity: For any $t=0,1,2, \ldots$,
       $C_t$ is connected.
 \item Single backbone: For any $t = 0,1,2 \ldots$,
       $B_t$ is connected.
 \item No interference: For any $t=0,1,2, \ldots$,
       the trajectories of two moving modules $m$ and $m'$
       never overlap. 
\end{enumerate}

The modules are anonymous and execute a common deterministic
distributed algorithm; they are \emph{uniform}. 
At each time step, each module observes the modules in its neighborhood 
and decides its movement; 
the modules are \emph{synchronous}. 
A cell $c_{i', j'}$ is a \emph{$k$-neighborhood} of cell $c_{i,j}$
if $|i'-i| \leq k $ and $|j'-j| \leq k$. 
A distributed algorithm of neighborhood size $k$ is a
total function that maps a 
$(2k+1) \times (2k+1)$ square grid to one cell; 
the modules are \emph{oblivious}. 
We assume that $k$ is constant with respect to $w$ and $h$, 
and a module can observe whether each cell in its $k$-neighborhood
is occupied by a module and whether it is a cell of a wall. 
When the modules are equipped with the \emph{global compass}, 
they share common north, south, east, and west directions.
When the modules are not equipped with the global compass, 
they do not know directions and their observations may be inconsistent.
However, we assume that the modules agree on the clockwise direction; 
they share a common handedness.

The \emph{state} of $R$ in $C_t$ is the local shape of $R$. 
We use $S^n$ to describe a state of $R$ consisting of $n$ modules. 
If the modules are equipped with the global compass,
the state of $R$ contains global directions, 
otherwise it does not contain any direction
because the modules cannot recognize any 
rotation on their state. 

When the modules are equipped with the global compass,
the execution of a given algorithm is uniquely determined 
by $C_0$ because the modules can agree on a total ordering 
among themselves. 
On the other hand, 
when the modules are not equipped with the global compass,
there exist multiple executions from $C_0$
depending on the local compass of each module. 

The \emph{search problem} requires 
the metamorphic robotic system to find a target put in one cell 
of a given field without any apriori information about 
the field, e.g., the size of the field, the target, 
the initial position of the metamorphic robotic system in the field, 
and so on. 
The target is stationary; it does not move during search. 
We call the cell where the target is put a \emph{target cell}. 
For simplicity, we assume that the field is large enough 
so that for any initial configuration there exists 
at least one cell that cannot be observed by any module. 
We say that the metamorphic robotic system 
\emph{finds} the target from a given initial configuration $C_0$
if in any execution from $C_0$, some module reaches the target cell 
and the metamorphic robotic system stops thereafter. 

The \emph{exploration problem} requires the metamorphic robotic system 
to visit all cells of the field without any apriori information 
about the field. 
We say that the metamorphic robotic system \emph{explores} 
a given field from a given initial configuration $C_0$ 
if in any execution from $C_0$, 
each cell of the field is visited by some module 
at least once.

\section{Search with global compass}
\label{sec:global}

In this section, we consider the metamorphic robotic system 
consisting of modules equipped with the global compass. 
When more than one modules equipped with the global compass 
form the metamorphic robotic system, 
they can agree on a total ordering of themselves. 
Hence, modules can perform different movements in any step. 
We give the following theorem.

\begin{thm}
 \label{theorem:global}
 Three modules equipped with the global compass 
 are necessary and sufficient for
 the metamorphic robotic system to find a target 
 in any given field from any initial configuration.  
\end{thm}

Theorem~\ref{theorem:global} indicates that 
three modules can perform self-stabilizing search 
in any given field because 
they can start from any initial configuration.

We demonstrate the necessity by impossibility for less than three modules 
in Section~\ref{subsec:global-nec} 
and the sufficiency with a search algorithm for three modules 
in Section~\ref{subsec:global-suf}.

\subsection{Impossibility for less than three modules} 
\label{subsec:global-nec}

By definition, a single module cannot perform any movement 
because there is no backbone module. 
When two modules form the metamorphic robotic system, 
one module can always perform a rotation. 
By repeating a rotation, 
the metamorphic robotic system can move in the field, 
however we demonstrate that it cannot find the target. 

\begin{lem} 
 \label{lemma:global-nec}
 Consider the metamorphic robotic system $R$ 
 consisting of less than three modules equipped with the global compass 
 in a sufficiently large field. 
 For any deterministic algorithm $A$ and any initial state of $R$, 
 there exists a choice of the target cell 
 such that $R$ cannot find the target. 
\end{lem} 
\begin{proof}
 The metamorphic robotic system consisting of a single module cannot 
 perform any movement and search is impossible.\footnote{If we allow
 random walk of a single module, of course, it can find the target
 with probability $1$.} 
 We consider the metamorphic robotic system $R$ 
 consisting of two modules,
 each of which can observe its $k(>0)$ neighborhood where
 $k$ is a constant with respect to the size of the field. 
 Assume that there exists a deterministic search algorithm $A$  
 for the two modules. 

 There are two connected states of the two modules. 
 In each state, $A$ specifies which module to take 
 what type of movement. 
 $R$ resumes the same state after two movements. 
 Figure~\ref{fig:2-move} shows 
 all possible moves from a horizontal state 
 after two steps. 
 See also Figure~\ref{fig:ex-2-move} as an example of 
 movements in two steps. 
 If the two modules do not see a cell or the target, 
 they repeat one of the eight moves and keep on moving 
 to one direction. 
 Note that we can obtain all possible moves after two movements 
 from the other (vertical) state by rotating 
 Figure~\ref{fig:2-move}. 

\begin{figure}[t]
 \centering
 \includegraphics[height=5cm]{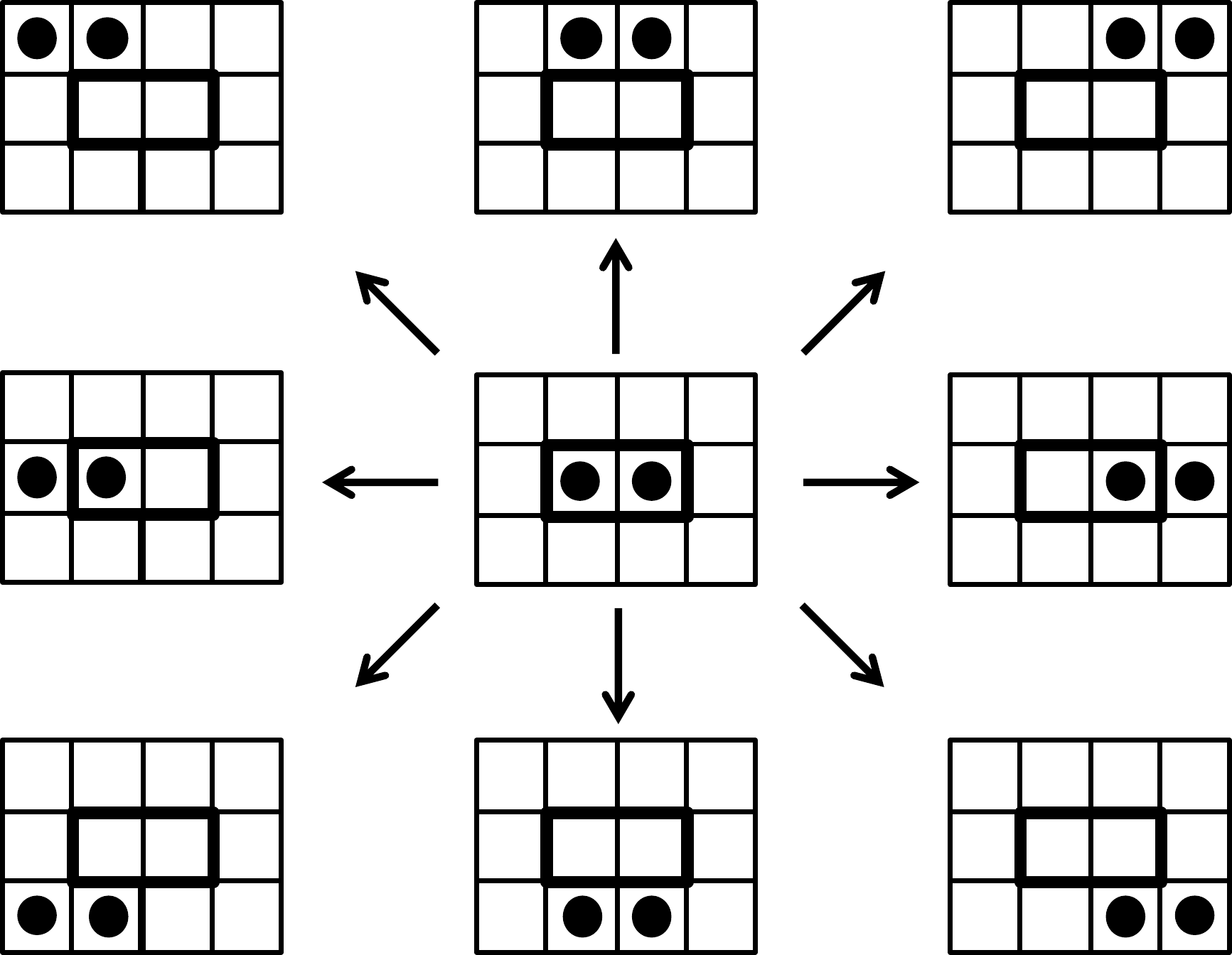}
 \caption{Possible moves of two modules. The bold rectangle marks 
 the initial positions of the two modules. } 
 \label{fig:2-move}
\end{figure}

\begin{figure}[t]
 \centering
 \includegraphics[width=6cm]{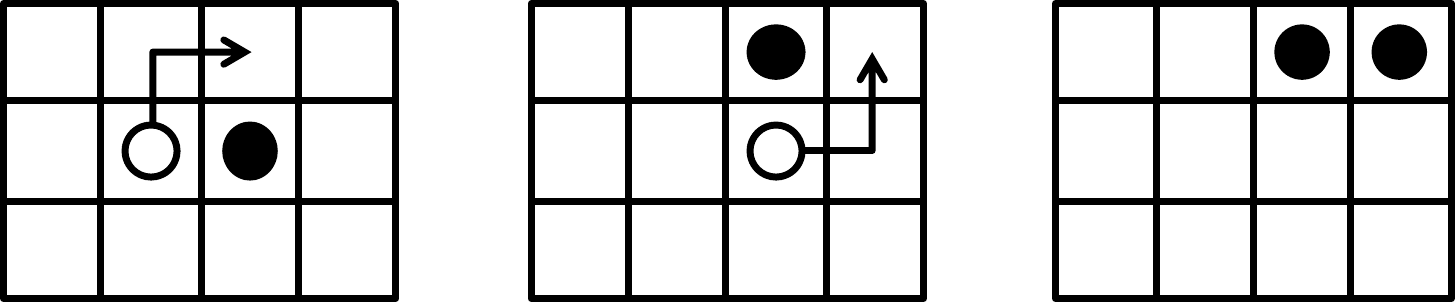}
 \caption{Example of two-step  moves.} 
 \label{fig:ex-2-move}
\end{figure}

 Without loss of generality, we assume that 
 the initial state of $R$ is the horizontal state. 
 We first consider the case where the modules start moving  
 to the northeast or to the east. 
 When the field is large enough so that 
 the modules around the center of the field
 cannot see the wall in their constant neighborhood, 
 $A$ makes $R$ keep on moving to one direction. 
 For any $w$ and $h$, there exists an initial configuration where 
 one module is in the $\lfloor w/2 \rfloor$th column 
 and after the straight move the two modules reach a wall where 
 they cannot see any corner of the field. 
 If the two modules leave the wall, 
 as soon as the wall goes out of their sight, 
 they return to a wall with the straight move. 
 The modules can visit the cells near the wall but they cannot leave the wall.  
 Another choice of $A$ is to make $R$ 
 keep on moving along the wall to one direction. 
 In both cases, the two modules eventually reaches a corner 
 of the field. 
 At this corner, $A$ has no choice other than to 
 make the two modules move around the adjacent wall 
 because even when they leave the corner, 
 they will return to the first wall due to the straight move of $A$. 
 Then, by the same discussion as the first wall, 
 two modules move along the second wall until it reaches 
 the second corner. 
 There are two choices at the second corner and 
 Figure~\ref{fig:decisiondiagram} shows the decision diagram. 
 There are two different routes depending on the choice of the 
 second corner. 
 As shown in Figure~\ref{fig:2-track-a} and 
 \ref{fig:2-track-b}, there are initial configurations from 
 which these two routes cannot 
 visit all cells in the field. 
 If the target is put in unvisited cells, 
 the metamorphic robotic system cannot find it. 

 By the same discussion, when the two modules start moving to the south east, 
 there are cells that the metamorphic robotic system cannot visit. 
 In this case, the decision diagram branches not at the second corner 
 but at the third corner. 
 We can then extend the proof where the two modules start moving 
 to the north east, to the north, to the south east, to the north, 
 or to the south. 

 Consequently, two modules cannot find a target put in an arbitrary 
 cell of an arbitrary field. 
 \qed

\begin{figure}[t]
 \centering 
 \includegraphics[width=13cm]{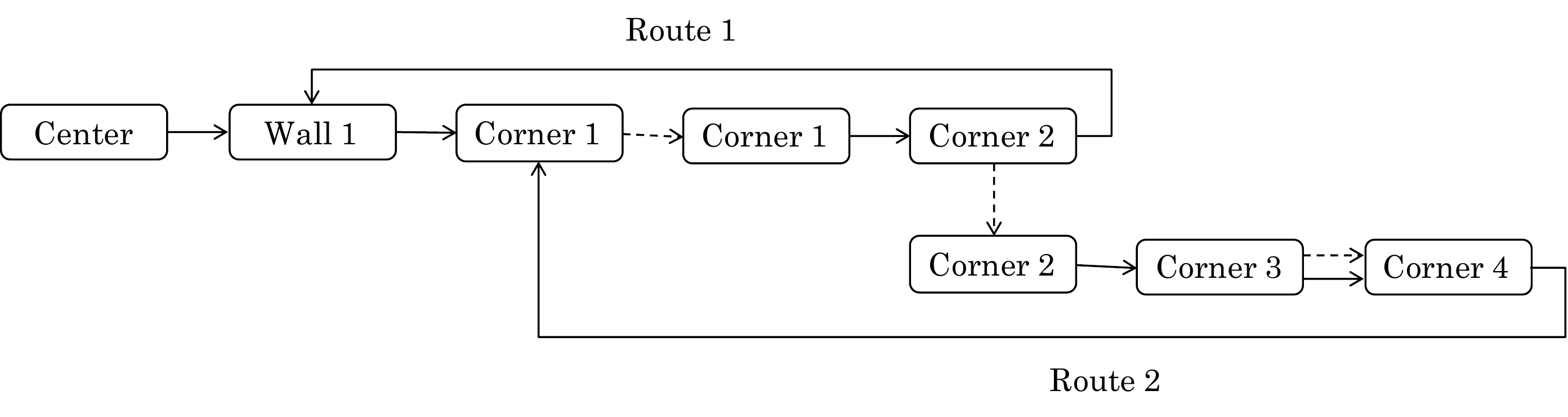}
 \caption{Decision diagram for the combination of moves when 
 the two modules starts to move to the northeast or to the east. 
 A solid arrow represents a straight move and a broken arrow represents a turn at a corner.}
 \label{fig:decisiondiagram}
\end{figure}

\begin{figure}[t]
 \centering 
 \begin{minipage}{0.45\hsize} 
  \centering
  \includegraphics[height=2.8cm]{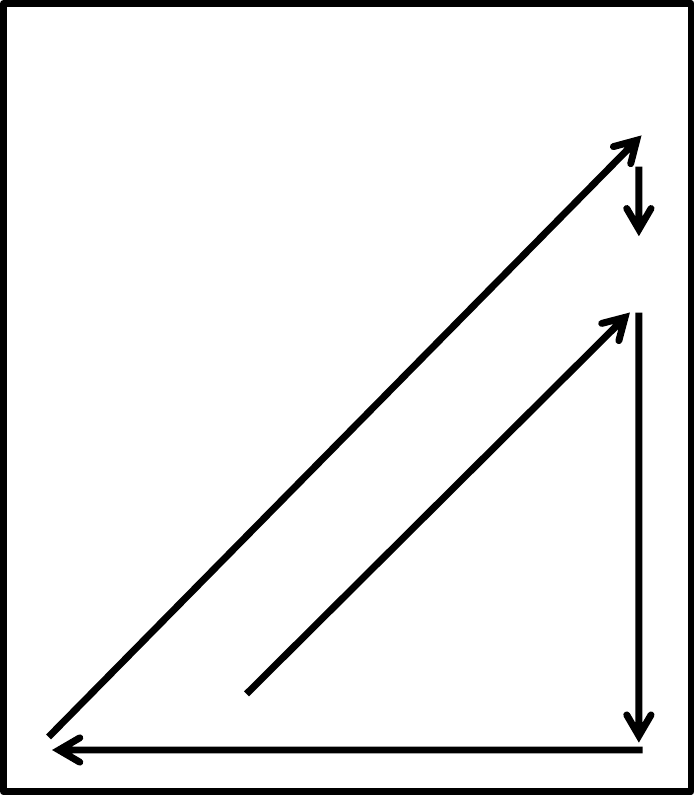}
  \caption{Example of Route $1$.}
  \label{fig:2-track-a}
 \end{minipage}
 \begin{minipage}{0.45\hsize} 
  \centering
  \includegraphics[height=2.8cm]{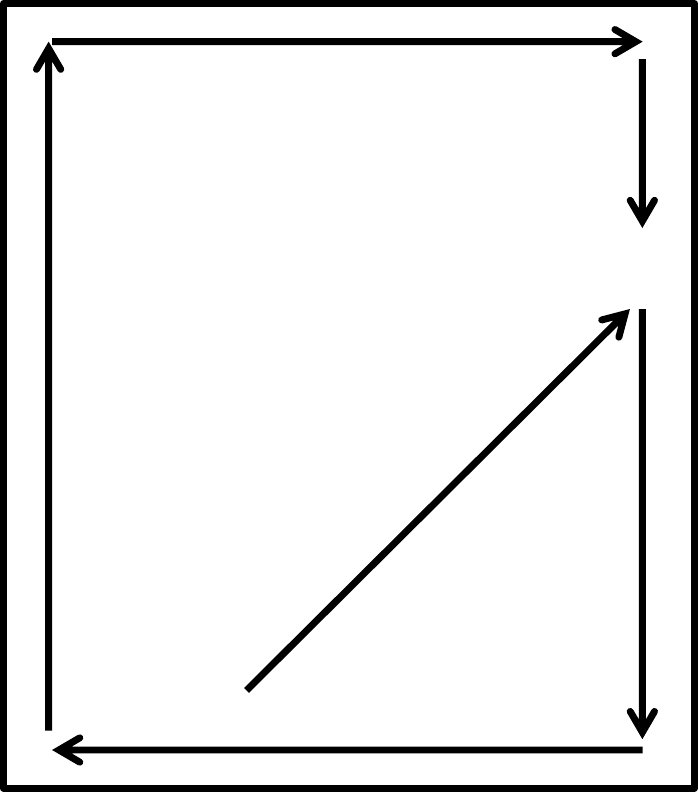}
  \caption{Example of Route $2$.}
  \label{fig:2-track-b}
 \end{minipage}
\end{figure}
\end{proof}

\subsection{Search algorithm for three modules} 
\label{subsec:global-suf}

In this section, we show the sufficiency of Theorem~\ref{theorem:global} 
with a search algorithm that enables 
three modules equipped with global compass to find a target. 
The proposed algorithm is so powerful that 
it can guarantee self-stabilizing search; 
three modules can start search from any initial configuration. 
We will show the following lemma. 
\begin{lem}
 \label{lemma:global-suf}
 Three modules equipped with the global compass 
 are sufficient for
 the metamorphic robotic system to find a target 
 in any given field from any initial configuration.  
\end{lem}

Our basic method is exploration, i.e., 
the metamorphic robotic system $R$ visits all cells of the field. 
$R$ moves to the south with sweeping each row.
However, since the initial configuration is arbitrary,
when it reaches the southernmost ($0$th) row,
it moves to the northernmost ($(h-1)$st) row
along either the east wall or the west wall
and it explores unvisited cells. 
Figure~\ref{fig:3-tracks} shows examples of ``tracks'' of $R$.
Depending on the number of rows and an initial configuration, 
$R$ moves along one of such tracks. 
We demonstrate the progress of exploration 
using a reference point of $R$ 
defined by its \emph{spine} and \emph{frontier} that will be defined
later.
The tracks in Figure~\ref{fig:3-tracks} show the tracks of
reference points. 
Note that the reference point does not refer to some specific module.
Rather, different modules serve as temporal reference points
during exploration. 

The proposed algorithm consists of the following basic moves;  
\begin{itemize}
 \item A move to the east and a move to the west.
 \item A turn on the east wall and a turn on the west wall. 
 \item A turn on the southwest corner and a turn on the southeast corner. 
 \item A move to the north wall along the east wall and
       that along the west wall. 
 \item A turn on the northeast corner and a turn on the northwest corner. 
\end{itemize}
Figure~\ref{fig:3-states} shows all possible states of $R$. 
We assume that each module can observe the cells in its
$2$-neighborhood. 
When one module of $R$ reaches a cell with the target,
$R$ stops. 
More precisely, due to the sufficient visibility, 
when one module reaches a target cell, 
the other modules can detect the target
and never perform any movement thereafter. 

\begin{figure}[t]
 \centering
   \includegraphics[height=2.8cm]{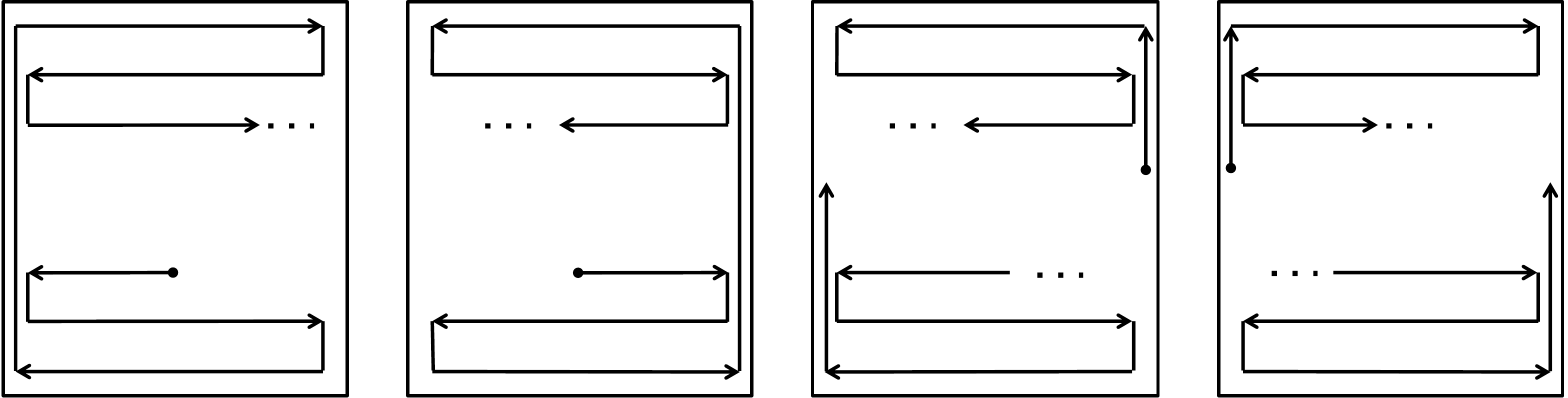}
 \caption{Example of exploration tracks of the metamorphic robotic system 
consisting of three modules.  
 Each track starts from the black circle.}
   \label{fig:3-tracks}
\end{figure}

\begin{figure}[t]
 \centering
 \includegraphics[height=2.5cm]{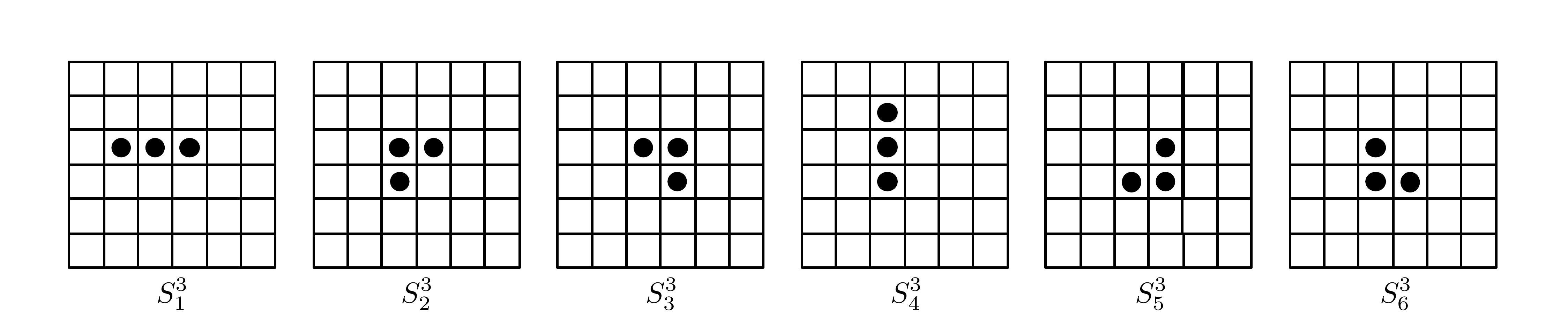}
 \caption{States of $R$ consisting of three modules} 
 \label{fig:3-states}
 \end{figure}

\begin{figure}[!htbp]
 \centering
 \begin{tabular}{c}

  \begin{minipage}{0.95\hsize}
   \centering 
   \includegraphics[height=2.5cm]{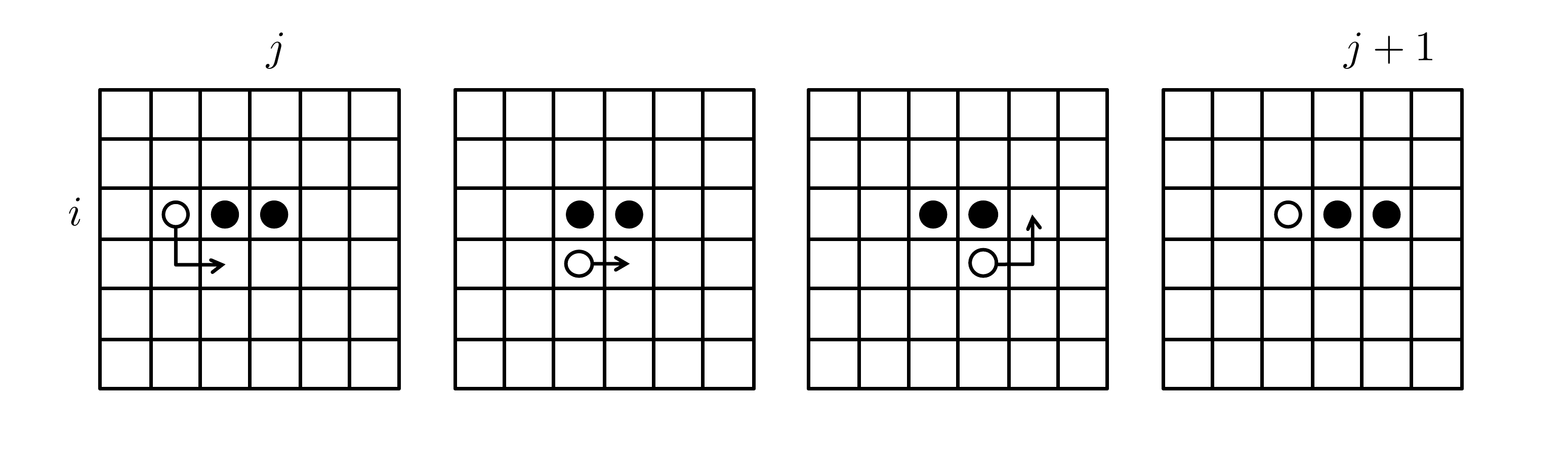}
   \caption{Move to the east ($S^3_1 \to S^3_2 \to S^3_3 \to S^3_1$)}
   \label{fig:3-to-east}
  \end{minipage}
\\ 
  \begin{minipage}{0.95\hsize}
   \centering 
   \includegraphics[height=2.5cm]{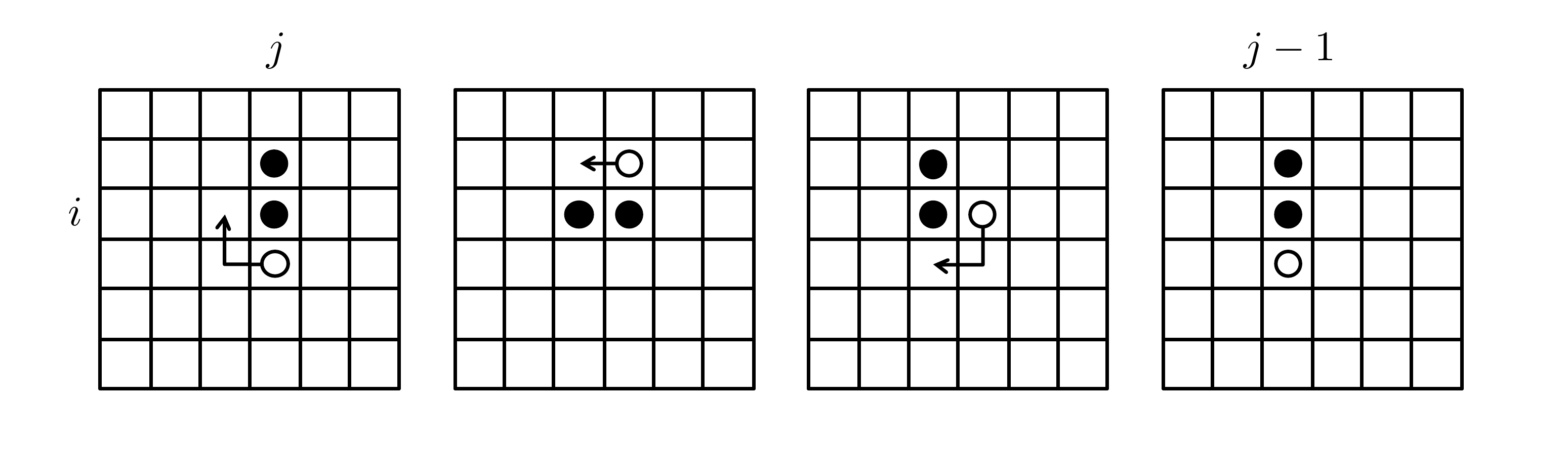}
   \caption{Move to the west ($S^3_4 \to S^3_5 \to S^3_6 \to S^3_4$)}
   \label{fig:3-to-west}
  \end{minipage}
  \\
    \begin{minipage}{0.95\hsize}
   \centering
   \includegraphics[height=2.5cm]{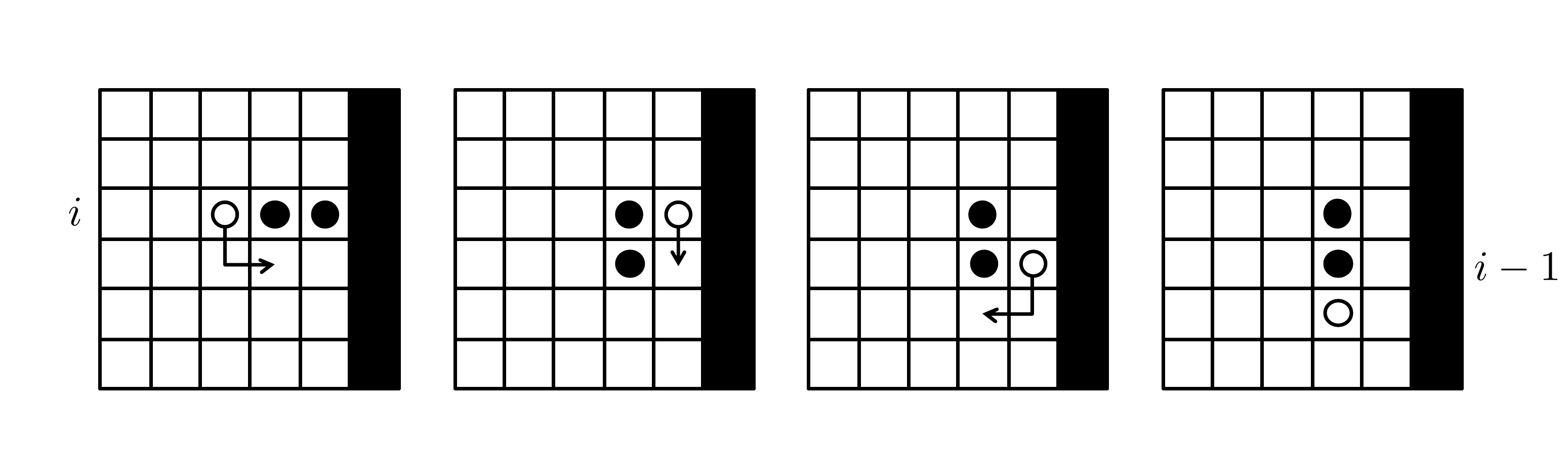}
   \caption{Turn on the east wall ($S^3_1 \to S^3_2 \to S^3_6 \to S^3_1$)} 
   \label{fig:3-east-wall}
  \end{minipage}
\\
  \begin{minipage}{0.95\hsize}
   \centering
   \includegraphics[height=2.5cm]{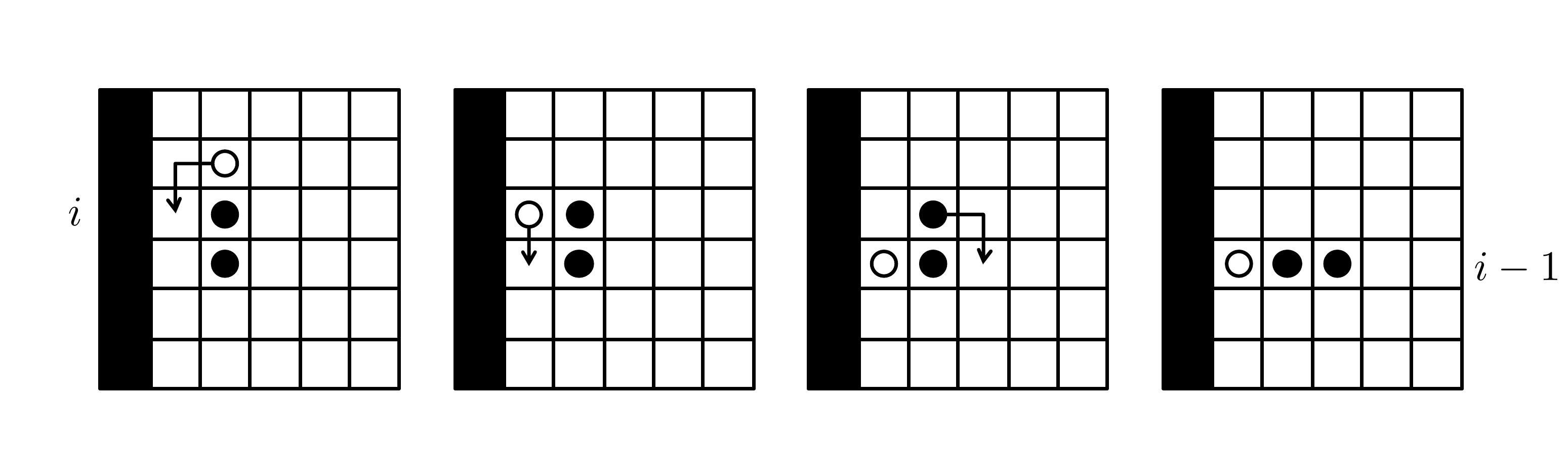}
   \caption{Turn on the west wall ($S^3_4 \to S^3_3 \to S^3_5 \to S^3_1$)} 
   \label{fig:3-west-wall}
  \end{minipage}
  \\ 
  \begin{minipage}{0.95\hsize}
   \centering
   \includegraphics[height=2.5cm]{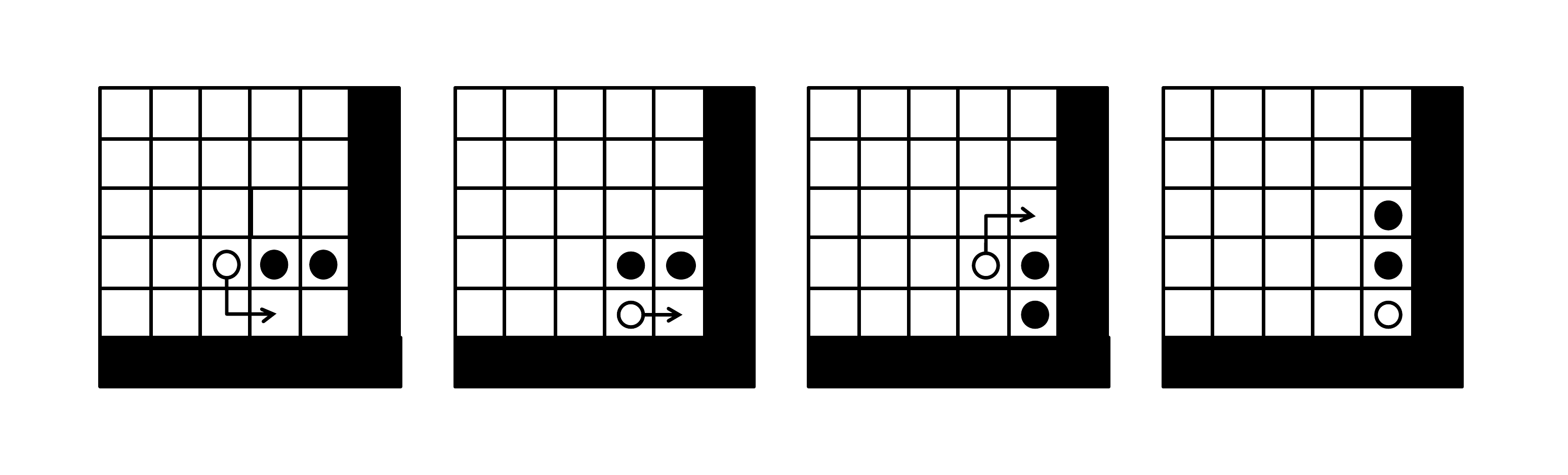}
   \caption{Turn on the southeast corner
    ($S^3_2 \to S^3_3 \to S^3_4$)} 
   \label{fig:3-south-east-corner}
  \end{minipage}
 \end{tabular}
\end{figure}

\begin{figure}[!htbp]
 \centering
 \begin{tabular}{c}
  \begin{minipage}{0.95\hsize}
   \centering
   \includegraphics[height=2.5cm]{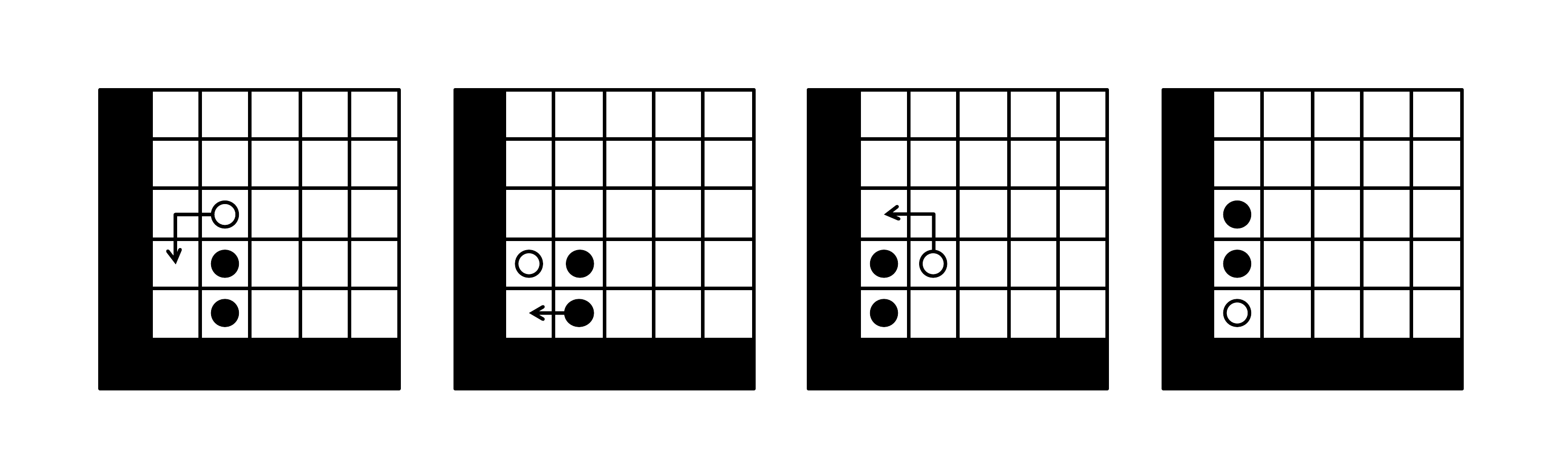}
   \caption{Turn on the southwest corner
   ($S^3_4 \to S^3_3 \to S^3_2 \to S^3_4$)} 
   \label{fig:3-south-west-corner}
  \end{minipage}
\\ 
  \begin{minipage}{0.95\hsize}
   \centering
   \includegraphics[height=2.5cm]{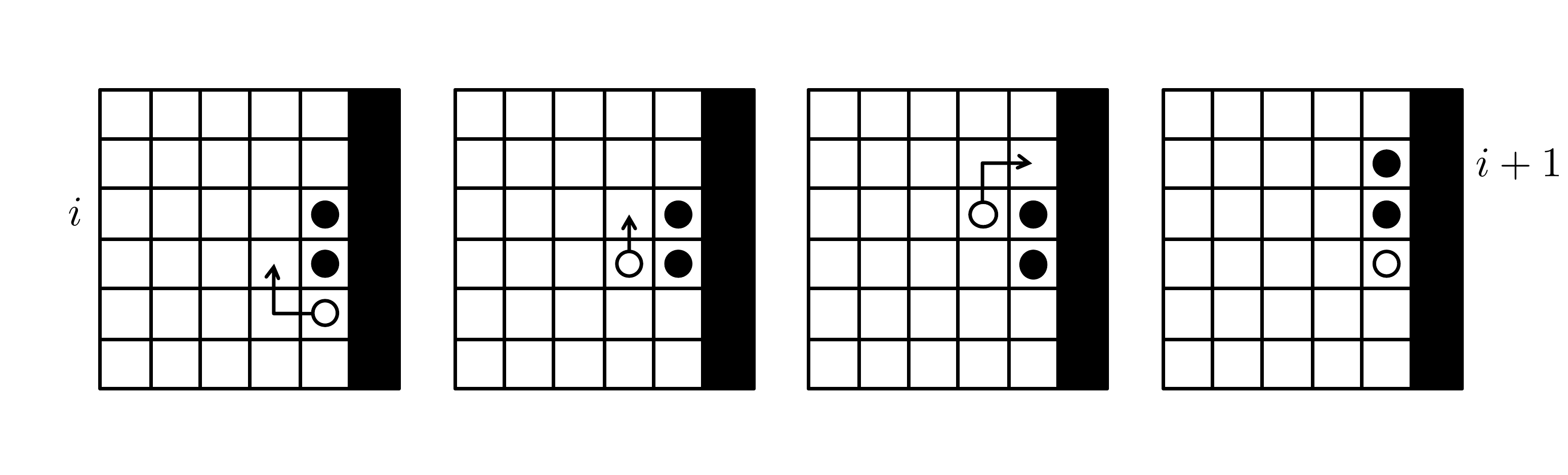}
   \caption{Move to the northeast corner
     ($S^3_4 \to S^3_5 \to S^3_3 \to S^3_4$)} 
   \label{fig:3-to-north-east-corner}
  \end{minipage}
  \\ 
  \begin{minipage}{0.95\hsize}
   \centering
   \includegraphics[height=2.5cm]{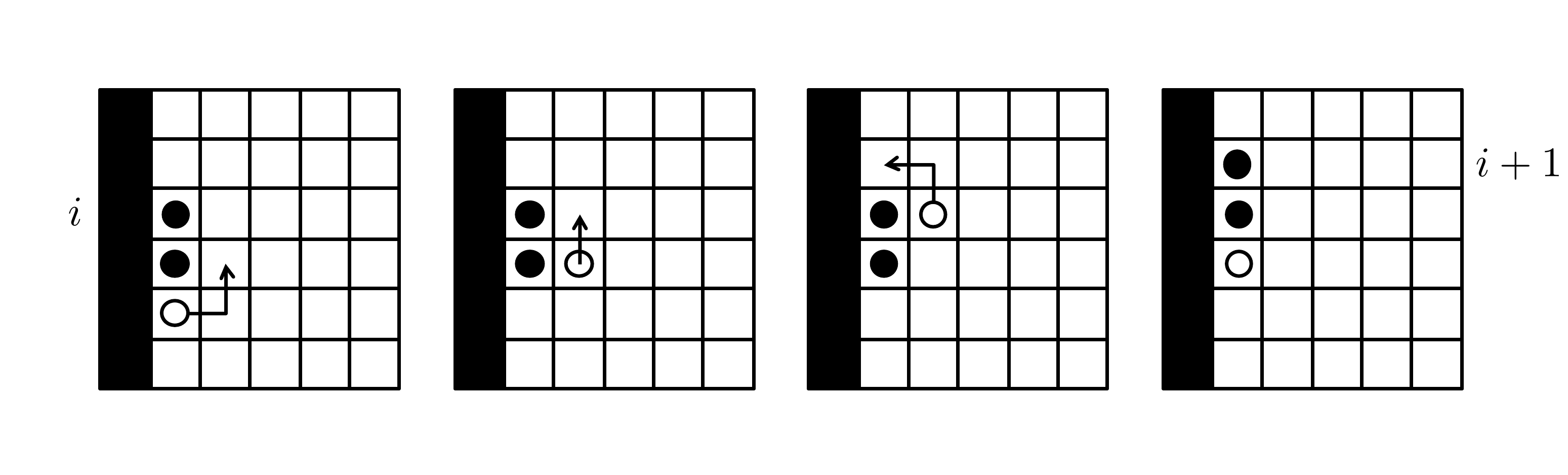}
   \caption{Move to the northwest corner
   ($S^3_4 \to S^3_6 \to S^3_2 \to S^3_4$)} 
   \label{fig:3-to-north-west-corner}
  \end{minipage}
  \\
  \begin{minipage}{0.95\hsize}
   \centering
   \includegraphics[height=2.5cm]{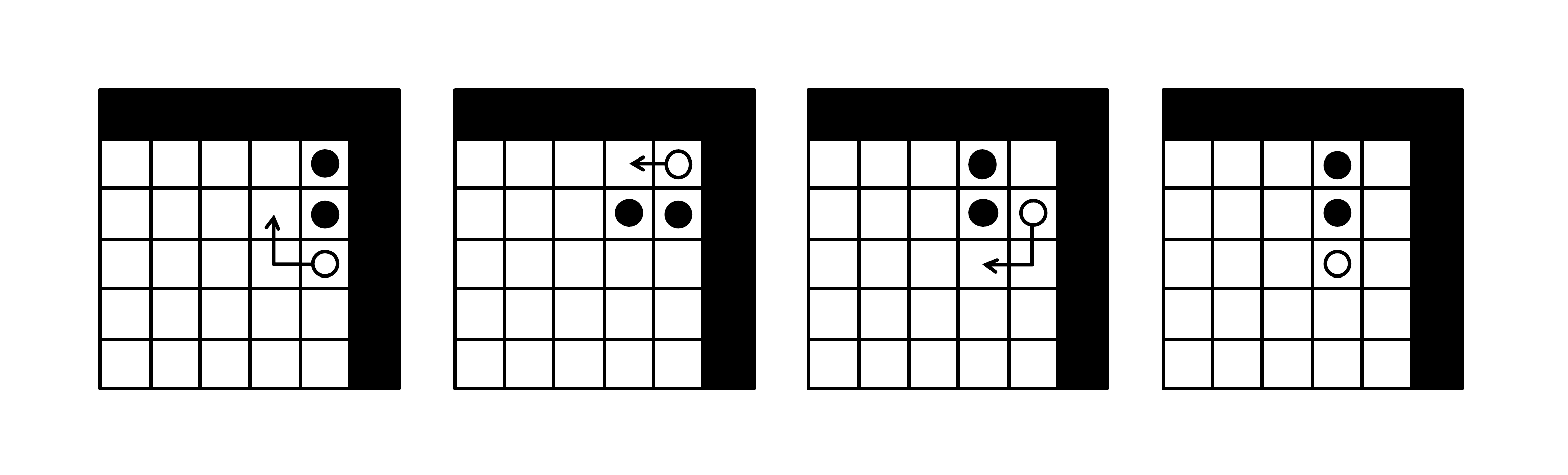}
   \caption{Turn on the northeast corner
   ($S^3_4 \to S^3_5 \to S^3_6 \to S^3_4$)} 
   \label{fig:3-north-east-corner}
  \end{minipage}
  \\
  \begin{minipage}{0.95\hsize}
   \centering
   \includegraphics[height=2.5cm]{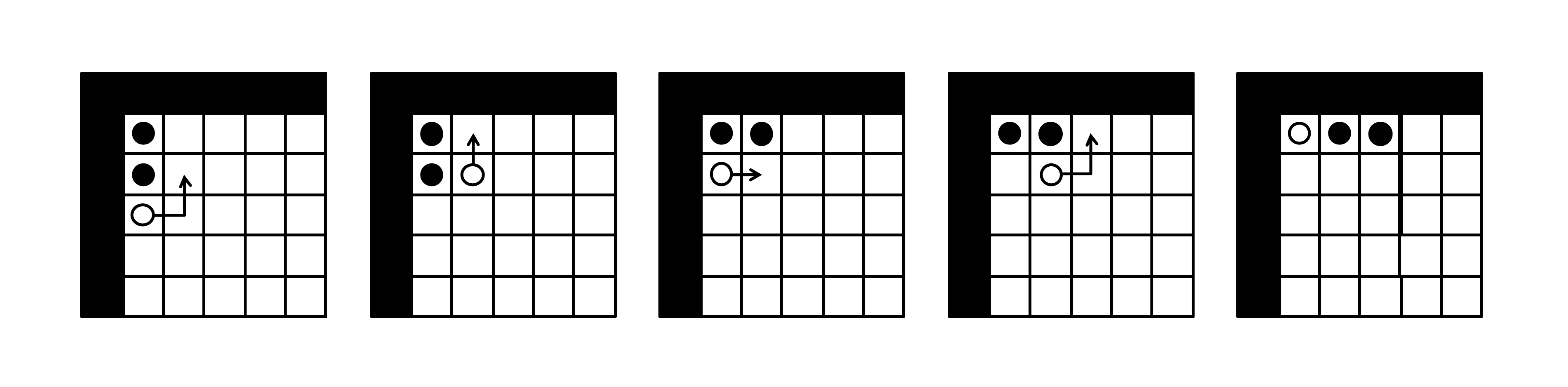}
   \caption{Turn on the northwest corner
   ($S^3_4 \to S^3_6 \to S^3_2 \to S^3_3 \to S^3_1$)} 
   \label{fig:3-north-west-corner}
  \end{minipage}
 \end{tabular}
\end{figure}

\begin{figure}[t]
 \centering
 \includegraphics[height=3cm]{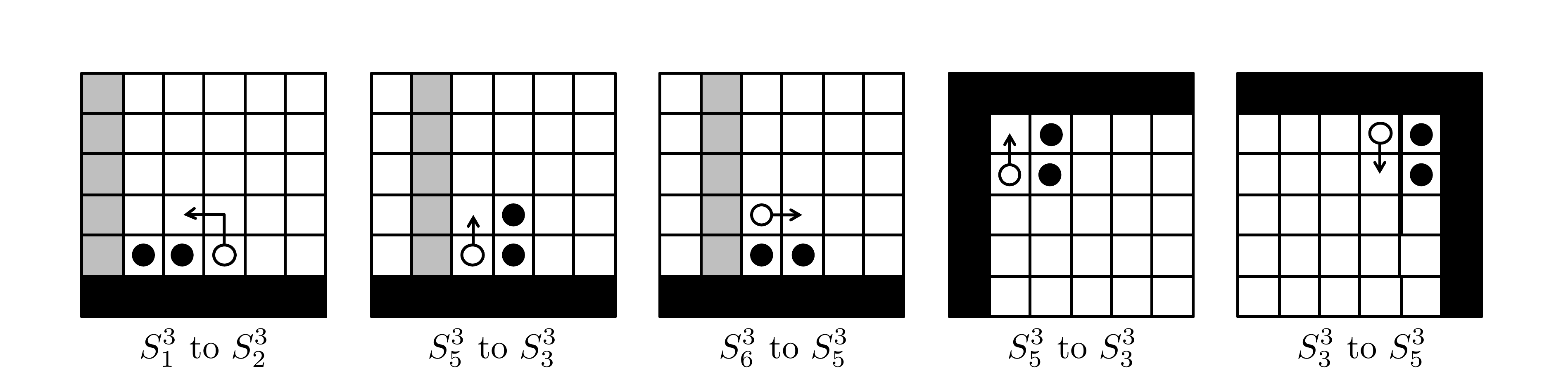}
 \caption{Additional movements. A gray column is either a wall or
 non-wall cells. } 
 \label{fig:3-exceptions}
\end{figure}

In the following, we show the proposed algorithm with a sequence of steps 
where just one module performs a rotation or sliding in each step. 
The proposed algorithm consists of ten types of basic moves. 
The point is that no pair of different moves use a state (including cells of walls )
in common. 
Hence, in each step, the modules can agree on what move they are performing and 
which module to perform a movement. 

\noindent{\bf Moves along a row.~} Figure~\ref{fig:3-to-east}  
and Figure~\ref{fig:3-to-west} show the move to the east and the move
to the west, respectively. 
By repeating one of the two moves, 
$R$ moves to one direction. 
Each module can observe the state of $R$ and
the two sets of states used in the two moves are disjoint. 
Thus, the modules can agree on the direction
to which $R$ is moving. 

In the first state of the unit move to the east,
the spine is the $i$th row and the frontier is
the $j$th column (Figure~\ref{fig:3-to-east}). 
At the end of a unit move,
the frontier reaches $(j+1)$st column.
During the move, the modules do not care whether
$(i+1)$st row, $(i-2)$nd row and $(j-3)$rd column are
walls or not. 

In the first state of the unit move to the west,
the spine is the $i$th row and the frontier is
the $j$th column (Figure~\ref{fig:3-to-west}). 
The modules do not care whether
$(i-2)$nd row, $(i+1)$st row and $(j+1)$st column are
walls or not. 

\noindent{\bf Turns on the walls.~} By repeating the above two moves, 
$R$ eventually reaches either the west wall or the east wall. 
Figure~\ref{fig:3-east-wall} 
and Figure~\ref{fig:3-west-wall} show a turn on the east wall and
a turn on the west wall, respectively. 
$R$ changes its spine and starts a new move to the opposite direction. 

In the first step of a turn on the east wall, 
the leftmost module cannot see the east wall and performs a rotation 
in the same way as the move to the east. 
After that, all modules can see the east wall and 
they perform the movements as shown in Figure~\ref{fig:3-east-wall}. 
After the turn, $R$ performs the move to the west. 

In a turn on the west wall, 
the three modules see the west wall from the first step. 
After the turn, $R$ performs the move to the east. 

\noindent{\bf Turns on the south corners and moves to the north wall.~} By
repeating the above four 
moves, $R$ eventually reaches the south wall. 
Then, it turns and moves along either the east wall or the west wall 
until it reaches the north wall. 
Figure~\ref{fig:3-south-east-corner} and
Figure~\ref{fig:3-south-west-corner} show these turns. 

In the first step of a turn on the southeast corner, 
the leftmost module cannot see the east wall and performs a rotation 
in the same way as the move to the east. 
After that, all modules can see the east wall and performs the movements 
of Figure~\ref{fig:3-south-east-corner}. 

In the first step of a turn on the southwest corner, 
the top module cannot see the south wall and performs a rotation 
in the same way as the turn on the west wall. 
After that, all modules can see the east wall and performs the movements 
of Figure~\ref{fig:3-south-west-corner}. 

Figure~\ref{fig:3-to-north-east-corner} and
Figure~\ref{fig:3-to-north-west-corner} show the moves
to the north along the vertical walls. 
When $R$ moves along the east wall (the west wall, respectively), 
the frontier is the northernmost module and 
the spine is the $(w-1)$th column (the $0$th column, respectively). 
By repeating one of the two moves, the reference point of $R$ 
eventually reaches the $(h-1)$st row. 

\noindent{\bf Turns on the north corners.~} By repeating
the move in Figure~\ref{fig:3-to-north-east-corner} or 
Figure~\ref{fig:3-to-north-west-corner}, 
$R$ eventually reaches the north wall.
Then, it turns in the corner 
and starts moving along a row 
with the moves shown in Figure~\ref{fig:3-to-east}  
and Figure~\ref{fig:3-to-west}. 
Figure~\ref{fig:3-north-east-corner} and 
Figure~\ref{fig:3-north-west-corner} shows these turns. 
In the first step of a turn on the northeast corner 
(the northwest corner, respectively), 
the bottom module cannot see the north wall and performs a rotation 
in the same way as the move to the north corner. 
After that, all modules can see the east wall and performs the movements 
of Figure~\ref{fig:3-north-east-corner} 
(Figure~\ref{fig:3-north-west-corner}, respectively). 

Figure~\ref{fig:3-exceptions} shows how to handle exceptions. 
When $R$ is in the center of the field, 
all states appear in the above moves and any move can be executed.
However, when $R$ is on a wall or in a corner,
moves for some states are not defined or impossible. 
We add exceptional movements to avoid deadlocks in these states. 

The reference point of $R$ visits all cells in each row except the
southernmost row and the northernmost row.
The cells of the southernmost row 
are visited by the modules under 
the spine when $R$ moves along the $1$st row. 
The cells of the northernmost row 
are visited by the modules over 
the spine when $R$ moves along the $(h-2)$nd row.
Each module is visited by some module and 
$R$ eventually find the target and stops thereafter. 
Consequently, we have Lemma~\ref{lemma:global-suf}. 

By Lemma~\ref{lemma:global-nec} and \ref{lemma:global-suf}, 
we obtain Theorem~\ref{theorem:global}.

\section{Search without global compass}
\label{sec:local}
  
In this section, we consider the metamorphic robotic system 
consisting of modules not equipped with the global compass. 
In this case, symmetry among the modules causes deadlocks. 
Figure~\ref{fig:sym2} shows an example with two modules. 
Each module can perform a rotation; 
however when the two modules have symmetric local compasses, 
they cannot distinguish each other and perform symmetric rotations 
simultaneously. Then, the backbone is not maintained. 
Each module does not know the compass of the other module, 
thus they cannot perform a rotation. 
Figure~\ref{fig:sym4} shows another example with rotations.
The four modules cannot move
because if one of them moves, the other three modules 
may also move. 
Figure~\ref{fig:sym5} shows an example with 
slidings.\footnote{Deadlock states with rotations 
are also illustrated in \cite{MSS19}.} 
Figure~\ref{fig:sym3} shows another type of deadlock state. 
If one endpoint module performs a rotation,
the other endpoint module also performs a rotation 
when they have symmetric local compasses with respect to their midpoint. 
The two modules perform symmetric movements simultaneously 
and the three modules cannot move forward. 
Consequently, when modules are not equipped with the global compass, 
search is generally impossible from an arbitrary
initial configuration. 

We first focus on the minimum number of modules for search 
from a predefined ``good'' initial states. 
We show the following theorem through Section~\ref{subsec:local-nec} and \ref{subsec:local-suf}. 
We show the necessity with impossibility for less than five modules
and the sufficiency with a search algorithm for five modules.

\begin{thm}
 \label{theorem:local}
 Five modules not equipped with the global compass 
 are necessary and sufficient for
 the metamorphic robotic system to
 find the target in any given field from allowed initial states. 
\end{thm}

We then consider the minimum number of modules for search 
from an arbitrary initial state. 
As a basic symmetry breaking problem, 
we first consider \emph{locomotion} to one direction 
from an arbitrary initial state. 
We present deadlock states for even number of modules and 
$(4k+1)$ modules ($k=1,2,\ldots$). 
Then, we demonstrate that seven modules can start locomotion 
from an arbitrary initial state. 
We conclude that seven modules are necessary and sufficient for 
search from an arbitrary initial configuration. 

 \begin{figure}[t]
  \centering
   \begin{minipage}{0.4\hsize}
    \centering
    \includegraphics[height=2cm]{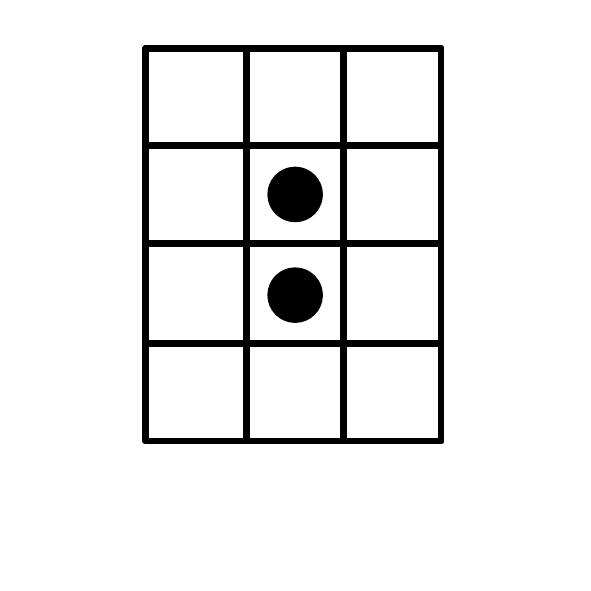}
    \caption{Two modules} 
    \label{fig:sym2}
   \end{minipage}
   \begin{minipage}{0.4\hsize}
    \centering
    \includegraphics[height=2cm]{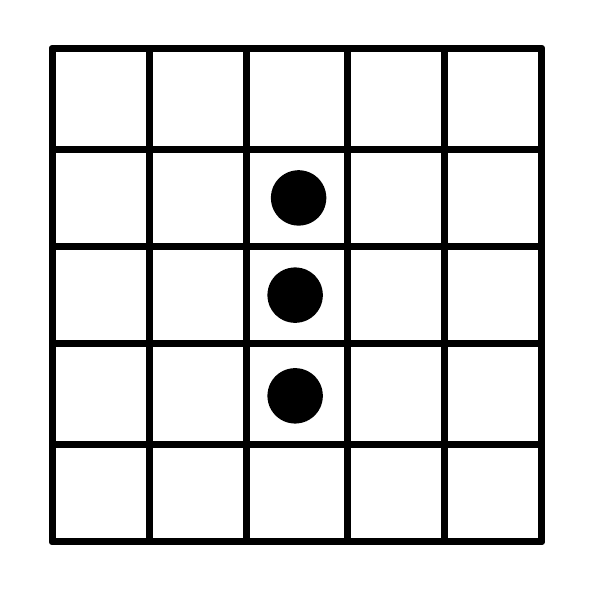}
    \caption{Three modules} 
    \label{fig:sym3}
   \end{minipage} 
\\ 
   \begin{minipage}{0.4\hsize}
    \centering
    \includegraphics[height=1.5cm]{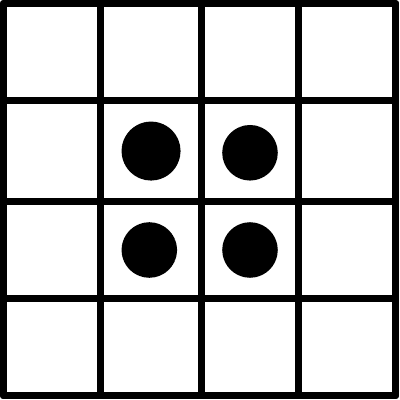}
    \caption{Four modules}
    \label{fig:sym4}
   \end{minipage}
   \begin{minipage}{0.4\hsize}
    \centering
    \includegraphics[height=2cm]{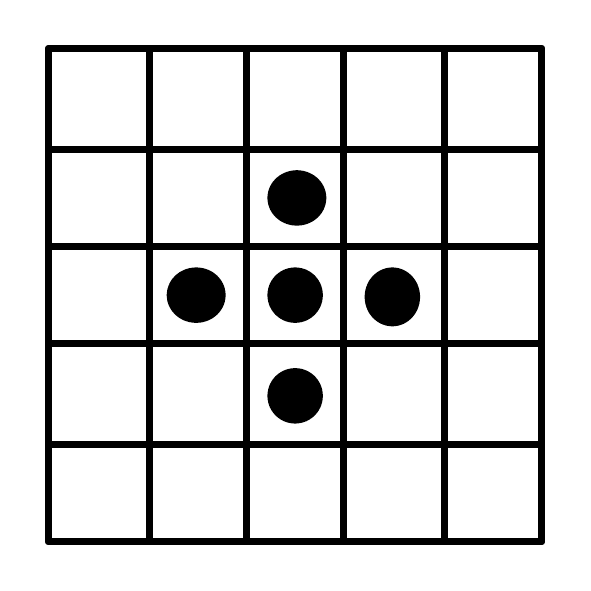}
    \caption{Five modules} 
    \label{fig:sym5}
   \end{minipage} 
 \end{figure}

\subsection{Impossibility for less than five modules} 
\label{subsec:local-nec} 

When the modules are not equipped with the global compass, 
there are two causes that make search by the metamorphic robotic system impossible. 
One is the deadlock in an initial state, 
and the other is the nonexistence of an exploration track 
as discussed in the proof of Lemma~\ref{lemma:global-nec}. 
In this section, we show the following lemma. 

\begin{lem} 
 \label{lemma:local-nec}
 Consider the metamorphic robotic system $R$ 
 consisting of less than five modules not equipped with the global compass 
 in a sufficiently large field.  
 For any deterministic algorithm $A$ and any initial state of $R$, 
 there exists a choice of the target cell 
 such that $R$ cannot find the target. 
\end{lem} 
\begin{proof}
We assume that each module can observe its $k$ neighborhood
where $k$ is a constant with respect to the size of the field. 
We have the following four cases.  

\noindent{\bf Case A: One module.~}
The single module cannot move because there is no backbone. 

\noindent{\bf Case B: Two modules.~}
There is a single initial connected state for the two modules.
Possible moves from the initial state are rotations, 
however if one module performs a rotation, 
the other module performs a symmetric rotation in the worst case. 
The two modules cannot maintain a backbone 
and they cannot move.  
 
\noindent{\bf Case C: Three modules.~} 
Consider the metamorphic robotic system $R^3$ consisting of 
three modules not equipped with the global compass. 
We can classify initial states into two types; 
the line states and the ``L'' states. 

\begin{figure}[t]
 \centering
 \begin{tabular}{c}

  \begin{minipage}{0.95\hsize}
   \centering
   \includegraphics[height=3cm]{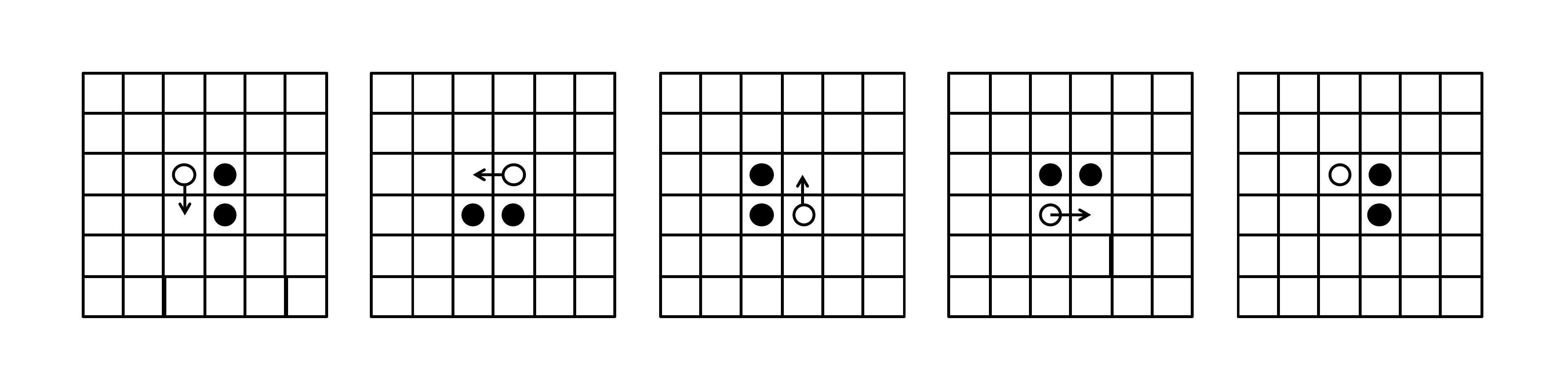}
   \caption{Sliding of the left module.}
   \label{fig:3-without-directions-left}
  \end{minipage}
  \\ 
  \begin{minipage}{0.95\hsize}
   \centering
   \includegraphics[height=3cm]{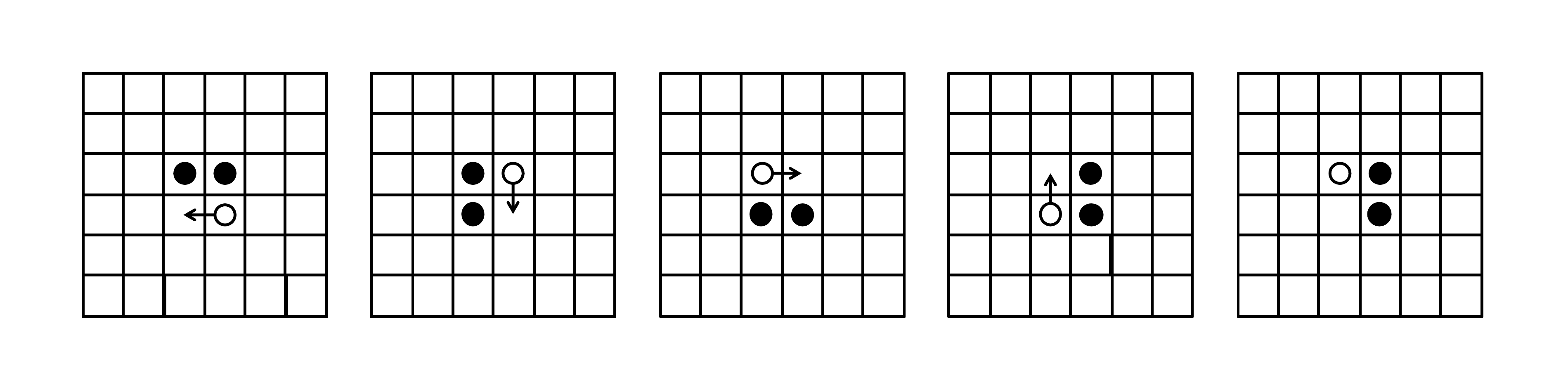}
   \caption{Sliding of the right module.}
   \label{fig:3-without-directions-right} 
  \end{minipage}
 \end{tabular}
\end{figure}

In a line state, the possible movements are rotations 
of endpoint modules. 
When one endpoint module performs a rotation, 
the other endpoint module also performs a rotation in the worst case. 
The new state is also a line, 
and the two endpoint modules repeat the same rotation. 
$R^3$ does not move forward to any direction. 
Thus, the line state should be avoided during the search procedure. 

In an ``L'' state, the possible movements are rotations
of the two endpoint modules. 
For the simplicity of the notation, we consider possible movements
from the first state of Figure~\ref{fig:3-without-directions-left} and 
Figure~\ref{fig:3-without-directions-right}. 

We call the left (upper) endpoint module \emph{the left module} and
the right (lower) endpoint module \emph{the right module}. 
Since the modules share the handedness,
they can agree on the left module and the right module. 
We have the following five cases. 

\noindent{Case C(i): The left module performs a rotation.~} 
The possible rotation is the clockwise rotation,
and the resulting state is the (vertical) line state. 
From the line state, $R$ cannot move forward to any direction 
in the worst case. 

\noindent{Case C(ii): The right module performs a rotation.~} 
The possible rotation is the counter-clockwise rotation,
and the resulting state is the (horizontal) line state. 
From the line state, $R$ cannot move forward to any direction 
in the worst case. 

\noindent{Case C(iii): Two endpoint modules perform rotations.~} 
When the both endpoint modules perform rotations, 
the resulting configuration is the L shape. 
Again, the two modules perform rotations,
and the resulting state is the initial state. 
$R$ cannot move forward to any direction. 

\noindent{Case C(iv): The left module performs a sliding.~} 
The possible sliding is the sliding to the south and the resulting state 
is again the L shape, where the other module becomes the new left 
module and performs a sliding.
Figure~\ref{fig:3-without-directions-left} shows the sequence 
of sliding movements, and after four steps, 
the configuration returns to the initial configuration. 
 
\noindent{Case C(v): The right module performs a slide.~} 
In the same way as Case C(iv), 
the configuration returns to the initial configuration. 
Figure~\ref{fig:3-without-directions-right} shows this case. 

Consequently, three modules cannot move forward to any direction 
and they cannot find a target that is out of the initial 
visibility ranges.

\begin{figure}[t]
 \centering
 \begin{tabular}{c}
  \begin{minipage}{0.95\hsize}
   \centering
   \includegraphics[height=6cm]{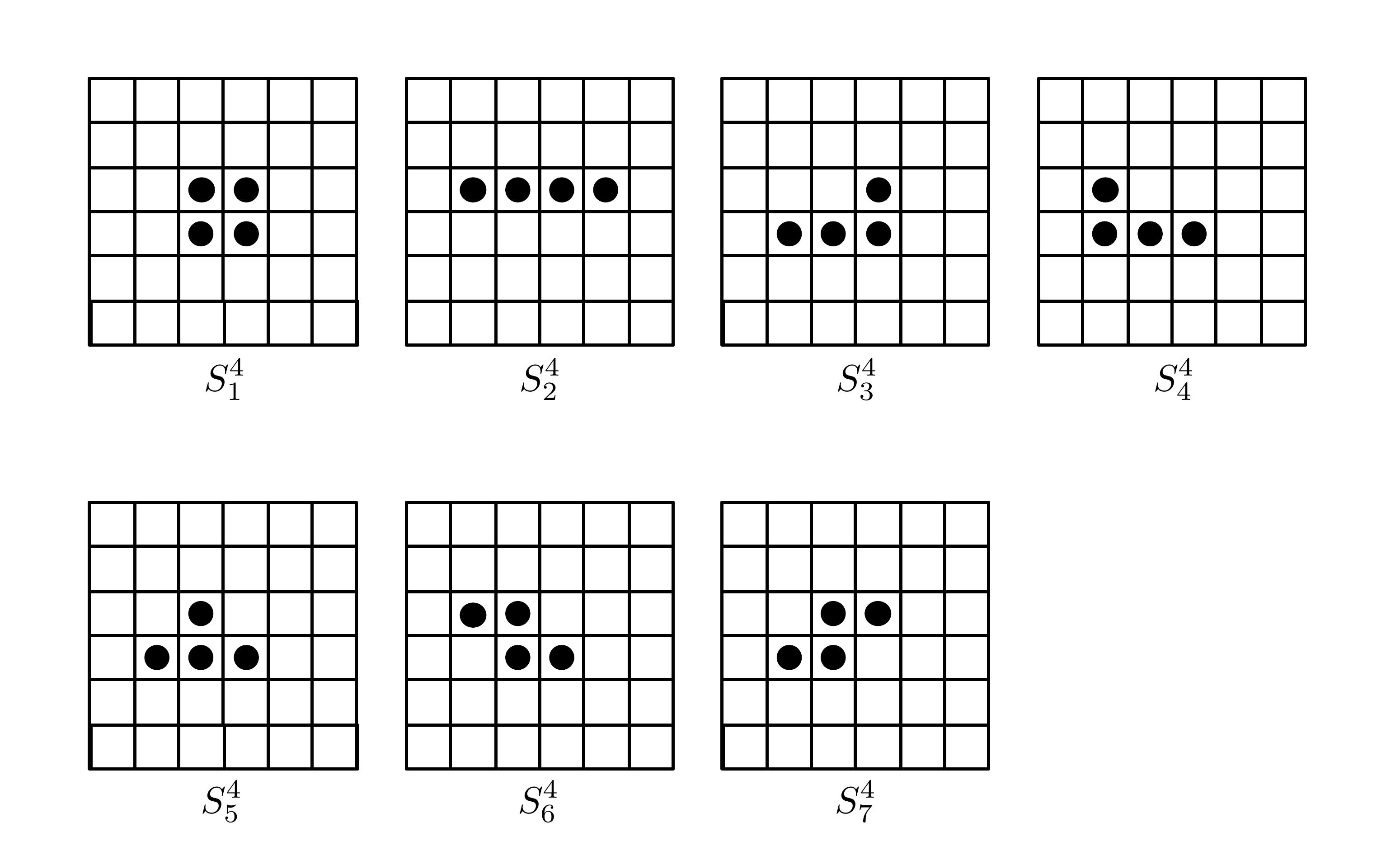}
   \caption{States of the metamorphic robotic system consisting of four modules.}
   \label{fig:4-states}
   \end{minipage}
  \\ 
  \\ 
  \begin{minipage}{0.95\hsize}
   \centering
   \includegraphics[height=4cm]{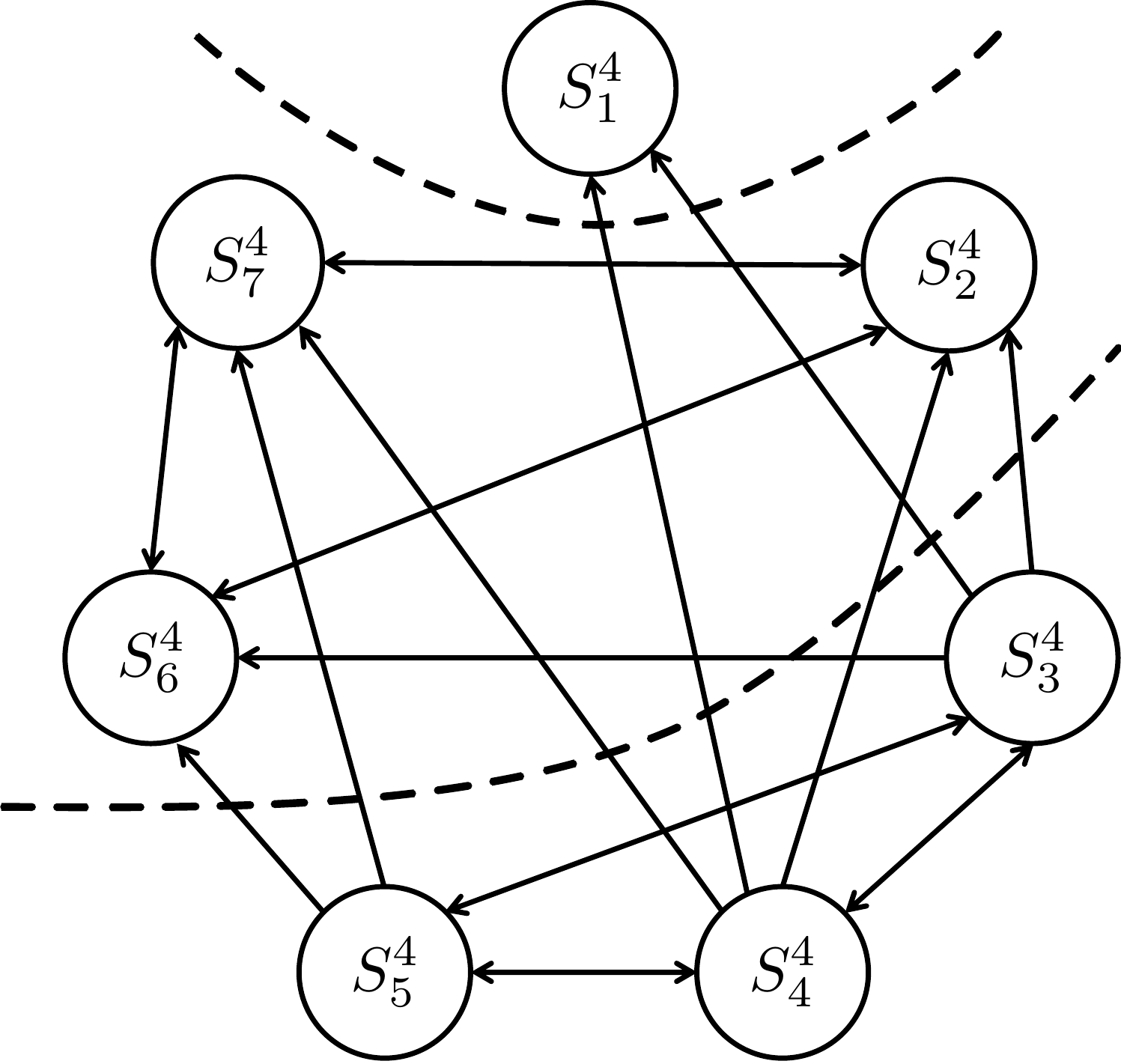}
   \caption{Reachability of the states}
   \label{fig:4-reachability}
  \end{minipage}
 \end{tabular}
\end{figure}

\noindent{\bf Case D: Four modules.~} 
Consider the metamorphic robotic system $R^4$ consisting of 
four modules not equipped with the global compass. 
Figure~\ref{fig:4-states} shows all possible states of $R^4$. 
There are seven states and Figure~\ref{fig:4-reachability}
shows the reachability of these states. 
An arc shows that there is a set of moves
that translates the state of its tail 
to the state of its head.
For example, $S^4_3$ is translated to $S^4_6$ 
by the rotation of the left module. 
State $S^4_6$ cannot be translated to $S^4_3$ because 
the two endpoint modules perform symmetric movements in the worst case. 
  
In $S^4_1$, no module can move because of symmetry,
and $S^4_2$, $S^4_6$, and $S^4_7$ form a loop
that has no outgoing arc. 
There are possible sequence of states consisting
of these three states (e.g.,  $S^4_2$, $S^4_6$, $S^4_7$, $S^4_2$); 
however, $R^4$ cannot move forward to any direction.
The remaining three states are $S^4_3$, $S^4_4$, and $S^4_5$. 
By repeating $S^4_3$ and $S^4_4$, $R^4$ can move
to one direction in the same way as Figure~\ref{fig:5-to-west}; 
however, neither the transition from $S^4_3$ to $S^4_5$ nor
$S^4_4$ to $S^4_5$ move $R^4$ forward.
The only way to move $R^4$ is to 
repeat $S^4_3$ and $S^4_4$. 
By turning $R^4$ upside-down,
it can move both directions; however, this just results in a loop
in two rows.
Hence, the four modules cannot find a target 
outside of the loop. 

Consequently, less than five modules cannot find a target 
even when we can choose an initial state of the metamorphic robotic system. 
\qed
\end{proof}

\subsection{Search algorithm for five modules} 
\label{subsec:local-suf}

In this section, we consider the metamorphic robotic system $R$ 
consisting of five modules not equipped with the global compass. 
Figure~\ref{fig:5-forbidden} shows the unique deadlock state 
$S^5_1$ and related three states $S^5_2$, $S^5_3$, and $S^5_4$ 
from which $R$ may transit to the deadlock state. 
Figure~\ref{fig:5-cycle} shows the transition diagram 
of these four states. 
Each arc represents the fact that there are possible movements 
that translates its starting state to its endpoint state. 
For example, in $S^5_4$ possible movements are
rotations of the two endpoint modules. 
However, when one of them moves,
the other may also move. 
Then, possible next states are $S^5_2$ and $S^5_3$. 
When two endpoint modules move in $S^5_2$ or $S^5_3$, 
possible next states are
$S^5_1$ (by $1$-slidings), $S^5_2$ (by $2$-slidings),
$S^5_3$ (by $2$-slidings), and $S^5_4$ (by rotations). 
However, $S^5_1$ cannot be translated to an asymmetric state 
in the worst case due to its symmetry. 
During these transitions, $R$ cannot move forward to any
direction. 
Consequently, these four states cannot be used in a search algorithm. 

\begin{figure}[t] 
 \centering
 \begin{tabular}{c}

  \begin{minipage}{0.95\hsize}
   \centering 
   \includegraphics[height=2.1cm]{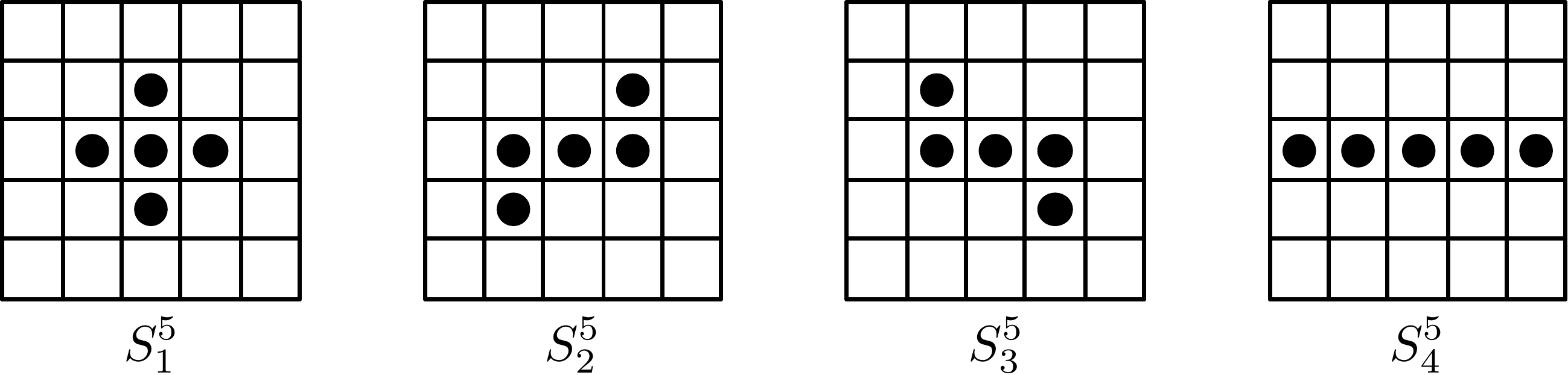}
   \caption{Forbidden states. }
   \label{fig:5-forbidden}
  \end{minipage}
  \\ 
  \\ 
  \begin{minipage}{0.95\hsize}
   \centering 
   \includegraphics[height=2.5cm]{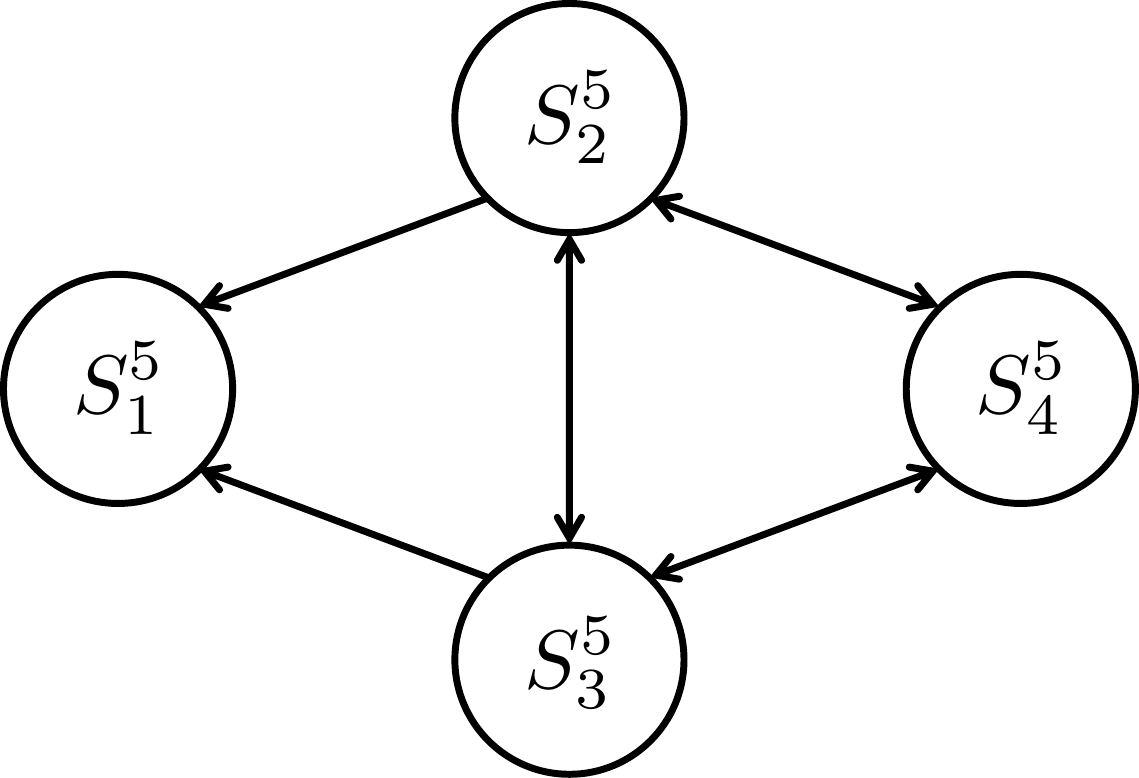}
   \caption{Cycle of the forbidden states. }
   \label{fig:5-cycle}
  \end{minipage}
 \end{tabular}
\end{figure} 

\begin{figure}[t]
 \centering
 \includegraphics[height=12cm]{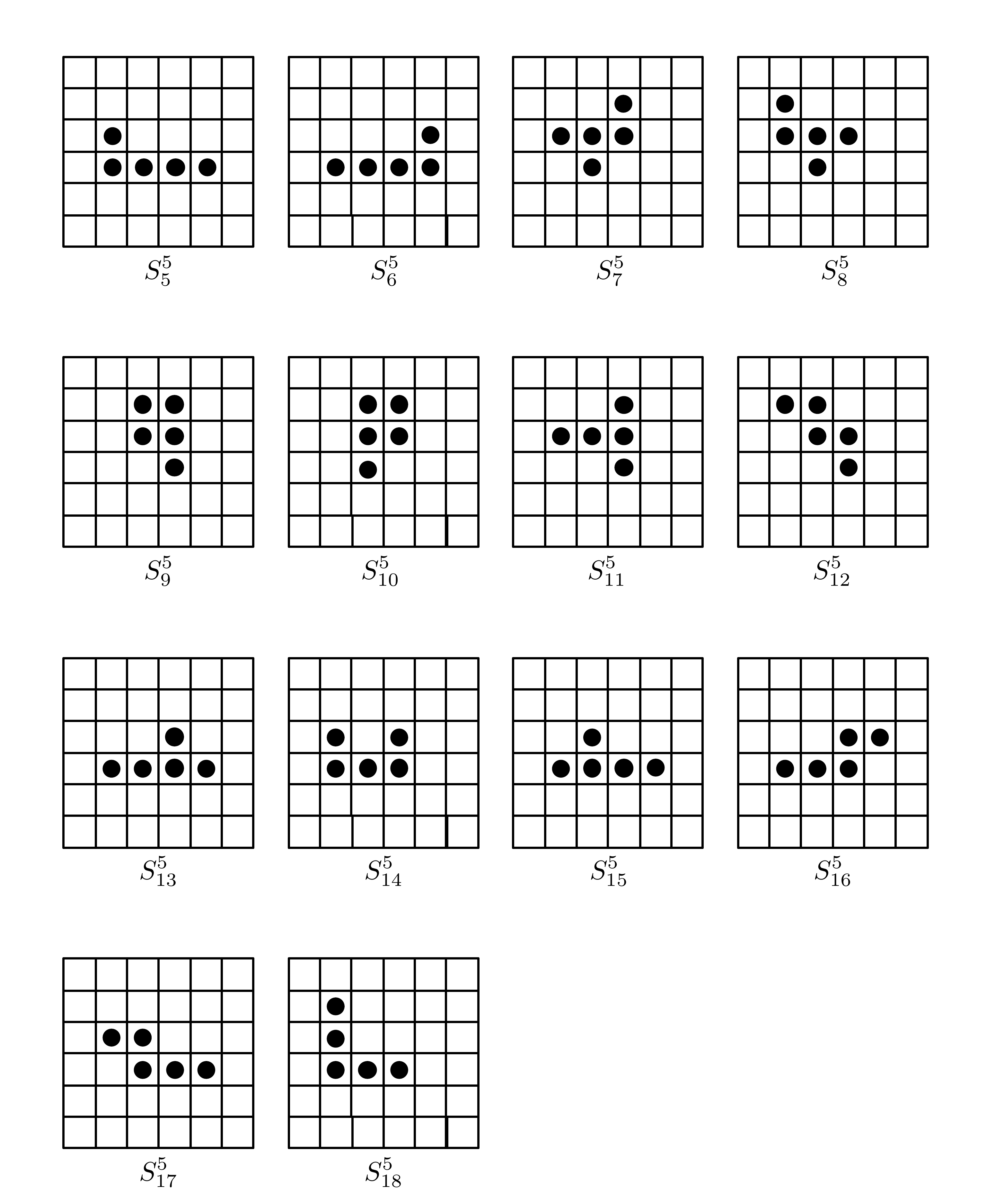}
 \caption{States of $R$ consisting of five modules} 
 \label{fig:5-states}
\end{figure}

\begin{figure}[!htbp]
      \centering 
 \includegraphics[height=2.6cm]{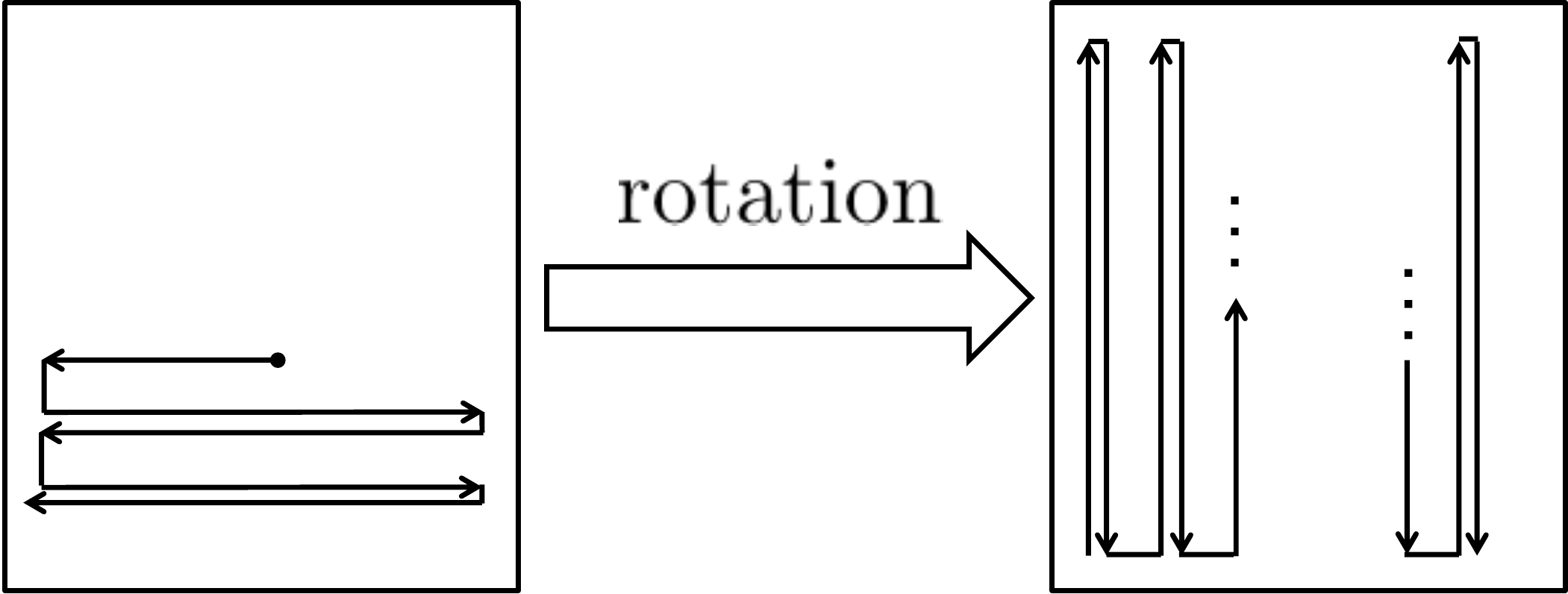}
 \caption{Exploration of the metamorphic
 robotic system consisting of five modules. } 
 \label{fig:5-tracks}
\end{figure}

In the following, we consider initial configurations where
the state of $R$ is none of the four states. 
Figure~\ref{fig:5-states} shows all the remaining states of $R$. 
Note that since the modules are not quipped with the global compass, 
they cannot recognize any rotation of a state. 
We assume that each module can observe the cells in its
$4$-neighborhood.
In addition, we use $2$-slidings and $3$-slidings,
which are not used in Section~\ref{subsec:global-suf}. 
To prove the following lemma, 
we present a search algorithm. 

\begin{lem}
 \label{lemma:local-suf}
 Five modules not equipped with the global compass 
 are sufficient for
 the metamorphic robotic system to find a target 
 in any given field from allowed initial states. 
\end{lem}

We adopt search by exploration in the same way as Section~\ref{subsec:global-suf}. 
However, the modules cannot use the global compass, and 
it is not easy to realize all the ten 
moves in Section~\ref{subsec:global-suf}.
Instead, $R$ uses a single track that checks
the rows from north to south 
with visiting each cell of a row from west to east 
(Figure~\ref{fig:5-tracks}). 
$R$ rotates the track by $\pi/2$ at the southwest corner 
in order to visit all cells.
It repeats the moves until it finds the target.
We explain the basic case where the directions 
are identical to the global compass. 

\begin{figure}[!htbp]
 \centering
 \begin{tabular}{c}

  \begin{minipage}{0.95\hsize}
   \centering
   \includegraphics[height=2.5cm]{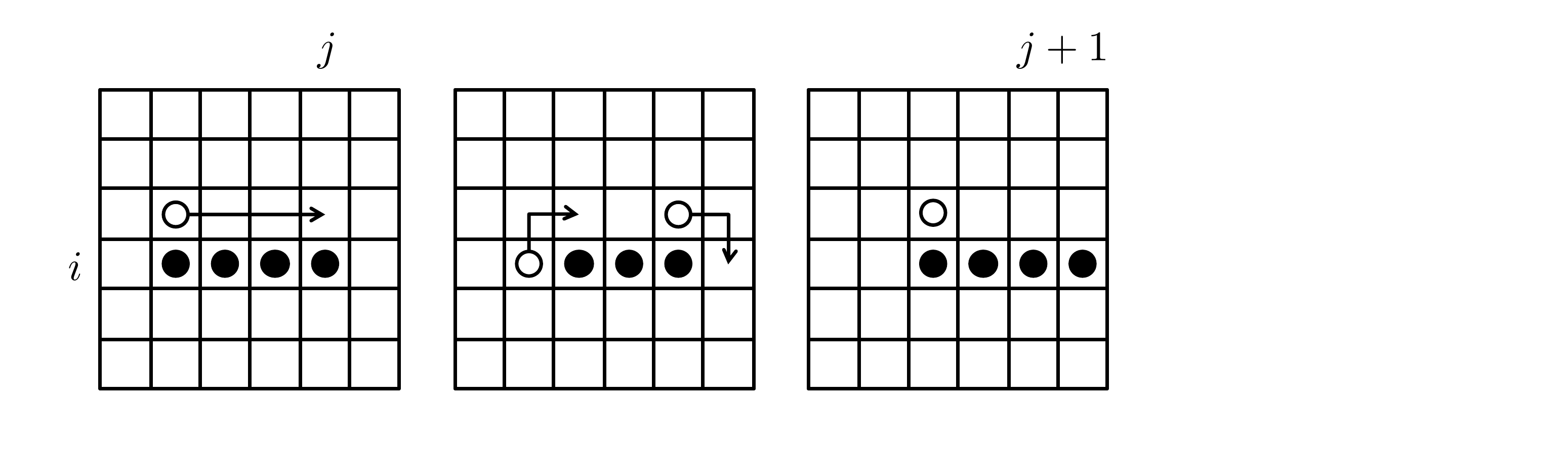}
   \caption{Move to the east ($S^5_5 \to S^5_6 \to S^5_5$)}
   \label{fig:5-to-east}
  \end{minipage}
  \\
    \begin{minipage}{0.95\hsize}
     \centering
     \includegraphics[height=2.5cm]{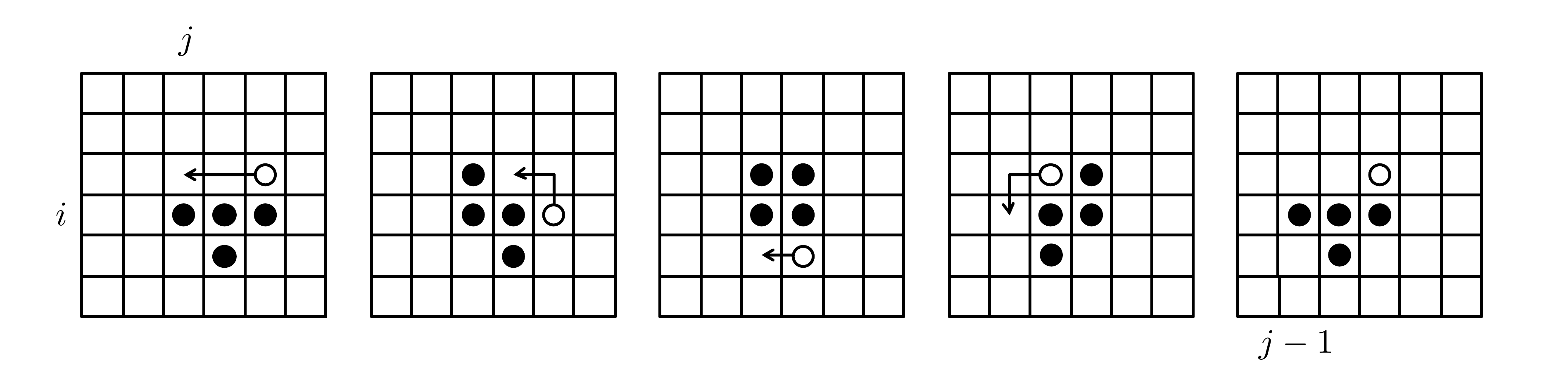}
     \caption{Move to the west
     ($S^5_7 \to S^5_8 \to S^5_9 \to S^5_{10} \to S^5_7$)}
     \label{fig:5-to-west}
    \end{minipage}
  \\
  \begin{minipage}{0.95\hsize}
   \centering
   \includegraphics[height=5cm]{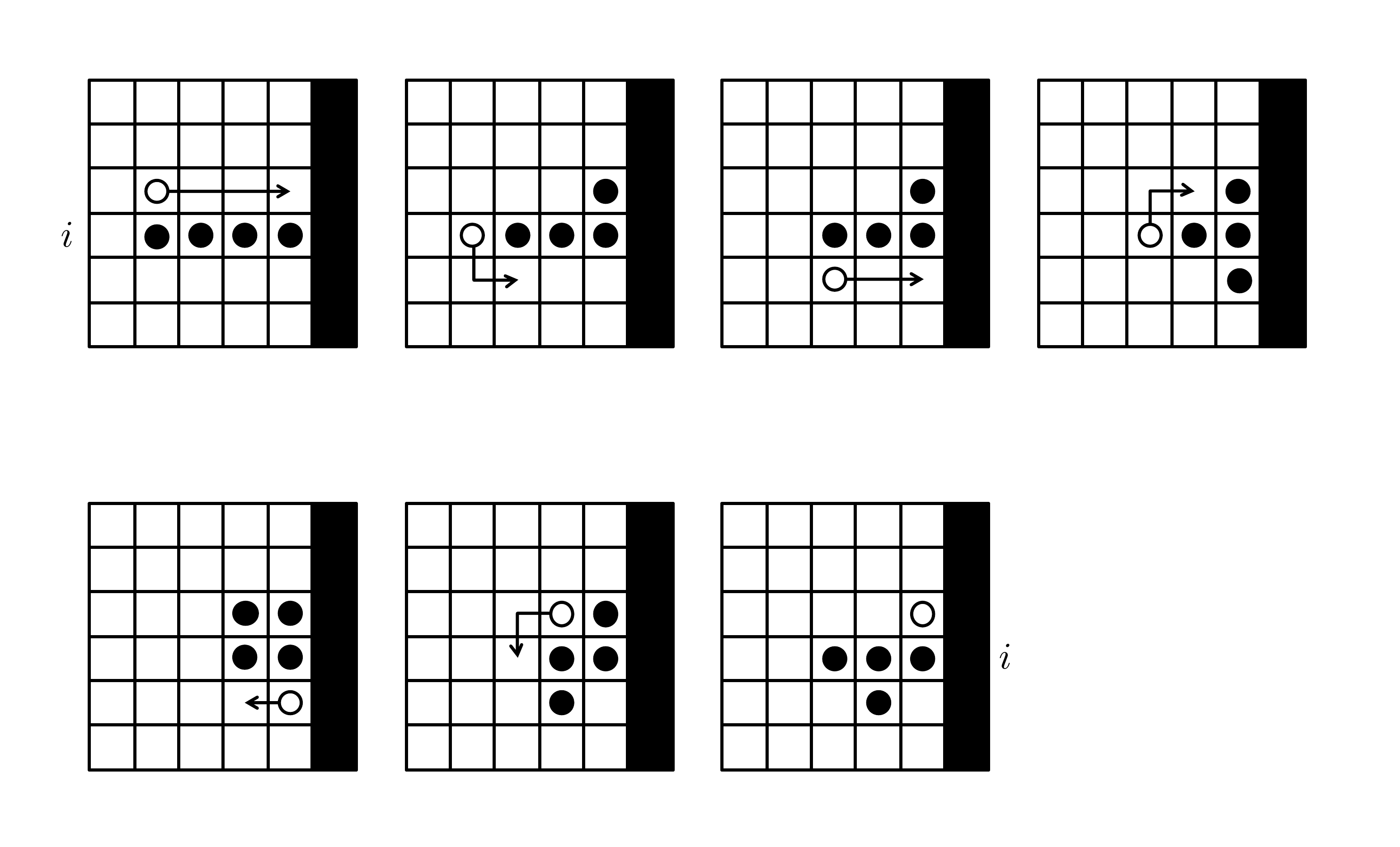}
   \caption{Turn on the east wall
   ($S^5_5 \to S^5_6 \to S^5_2 \to S^5_{11} \to S^5_9 \to S^5_{10} \to S^5_7$)}
   \label{fig:5-east-wall}
  \end{minipage}
  \\
  \begin{minipage}{0.95\hsize}
   \centering
   \includegraphics[height=5cm]{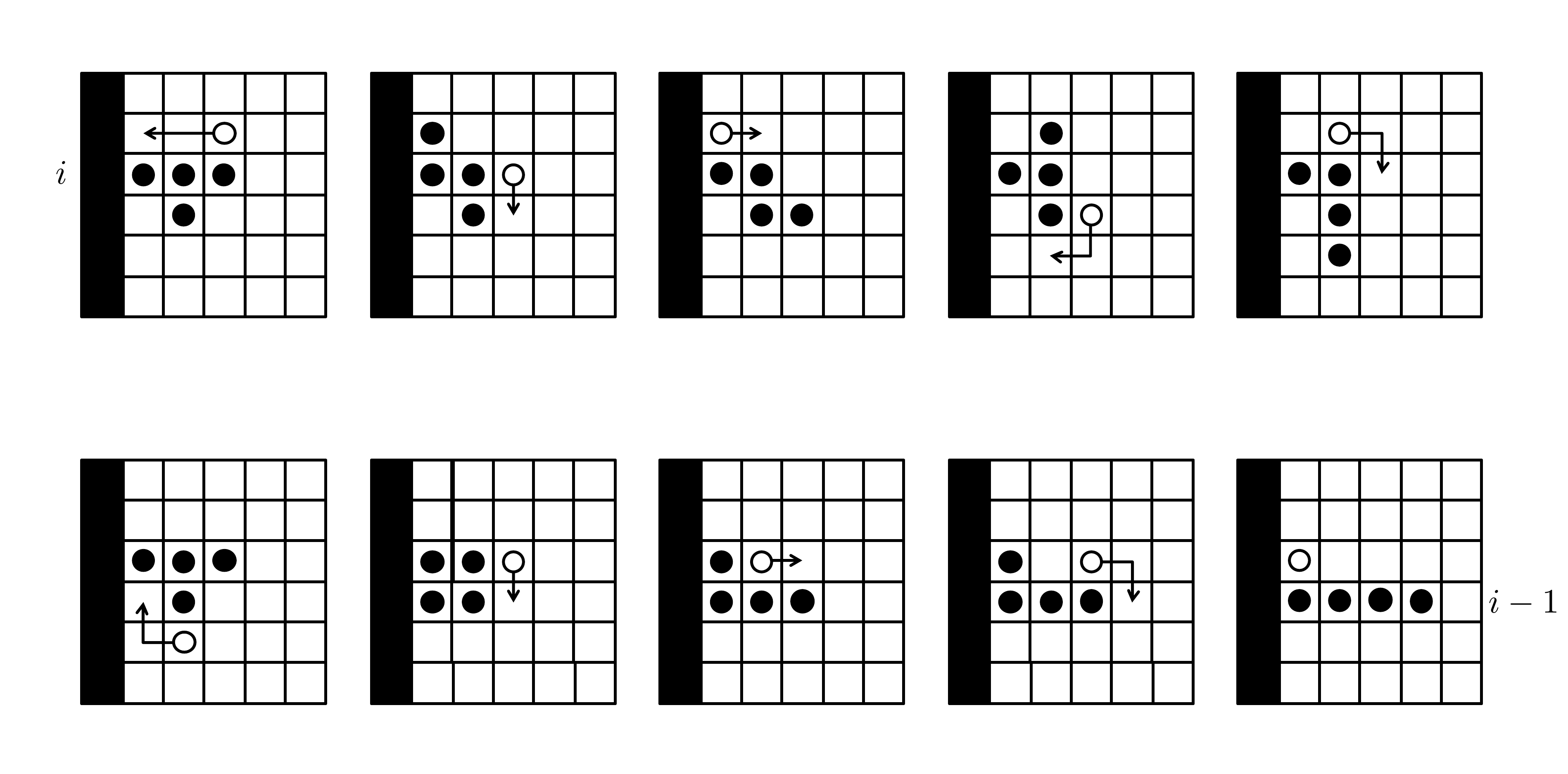}
   \caption{Turn on the west wall
   ($S^5_7 \to S^5_8 \to S^5_{12} \to S^5_7 \to S^5_{13} \to S^5_{11} \to
   S^5_9 \to S^5_{10} \to S^5_{14} \to S^5_5$)}
   \label{fig:5-west-wall}
  \end{minipage}

 \end{tabular}
\end{figure}

\begin{figure}[!htbp]
 \centering
 \begin{tabular}{c}
  \begin{minipage}{0.95\hsize}
   \centering
   \includegraphics[height=5cm]{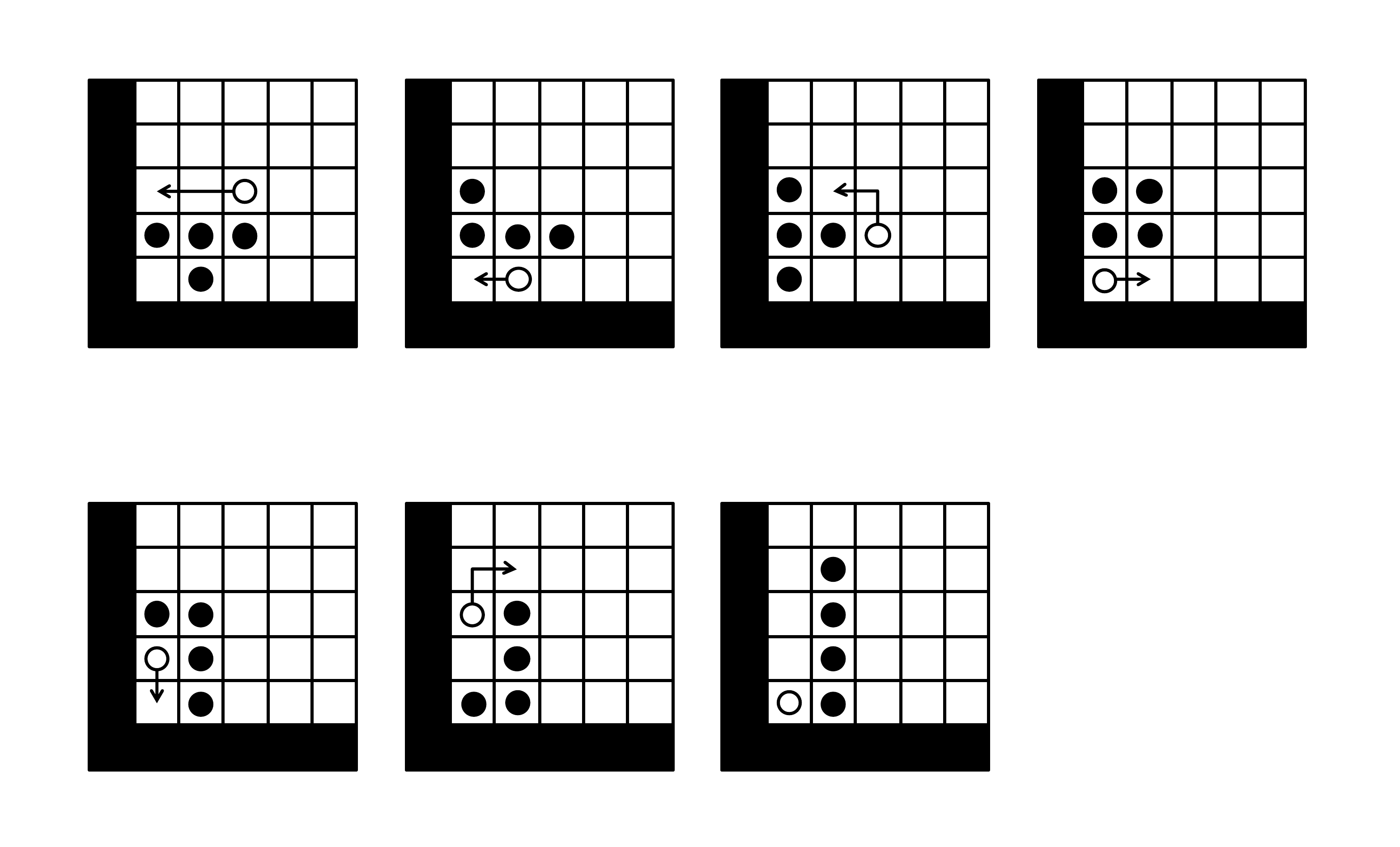}
   \caption{Turn on the southeast corner
   ($S^5_7 \to S^5_8 \to S^5_{11} \to S^5_{10} \to S^5_9 \to S^5_{14} \to
   S^5_5$)} 
   \label{fig:5-east-south-corner}
  \end{minipage}
  \\ 
  \begin{minipage}{0.95\hsize}
   \centering
   \includegraphics[height=6cm]{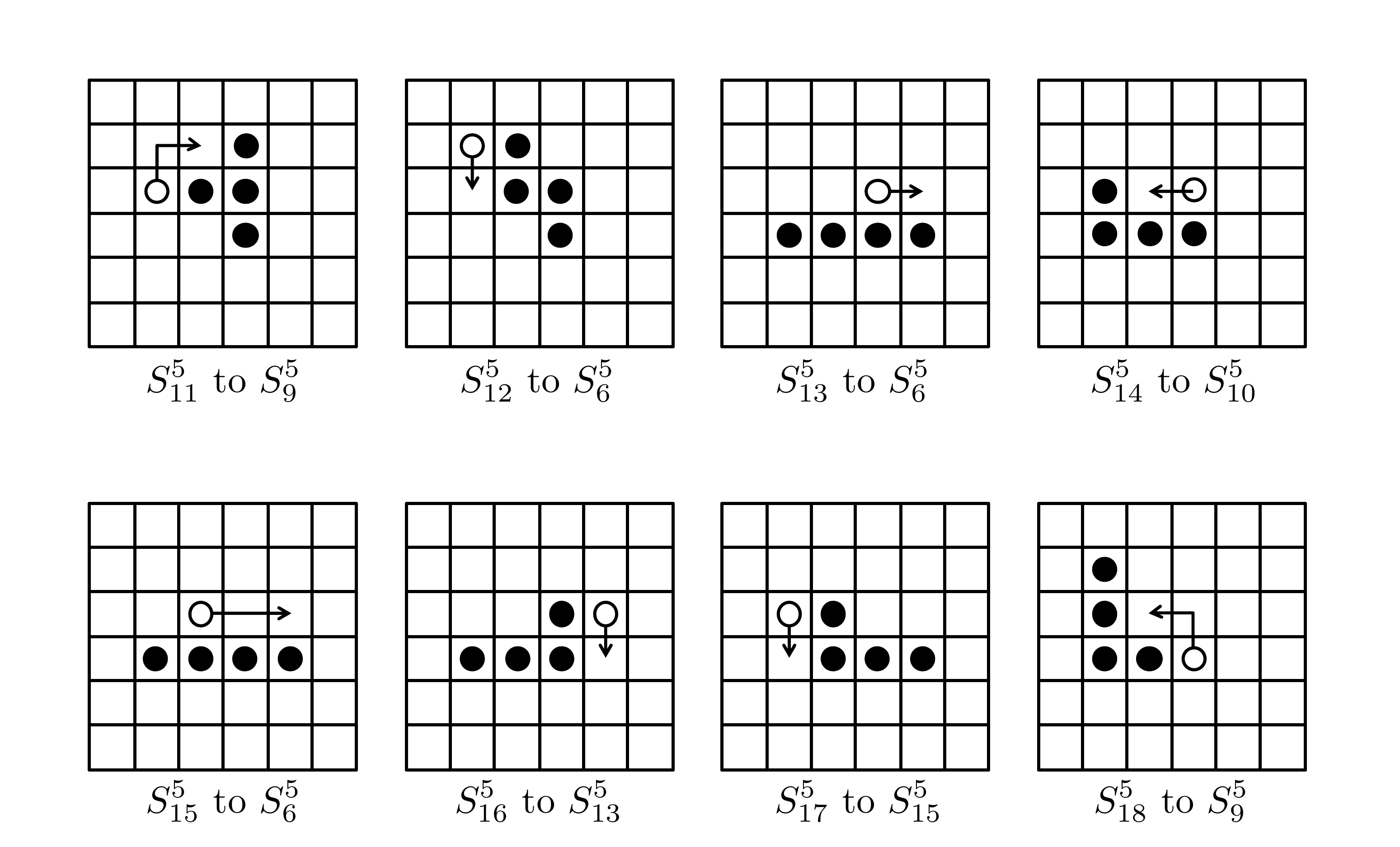}
   \caption{Exceptions without walls} 
   \label{fig:5-moving-exceptions}
  \end{minipage}
  \\ 
  \begin{minipage}{0.95\hsize}
   \centering
   \includegraphics[height=3cm]{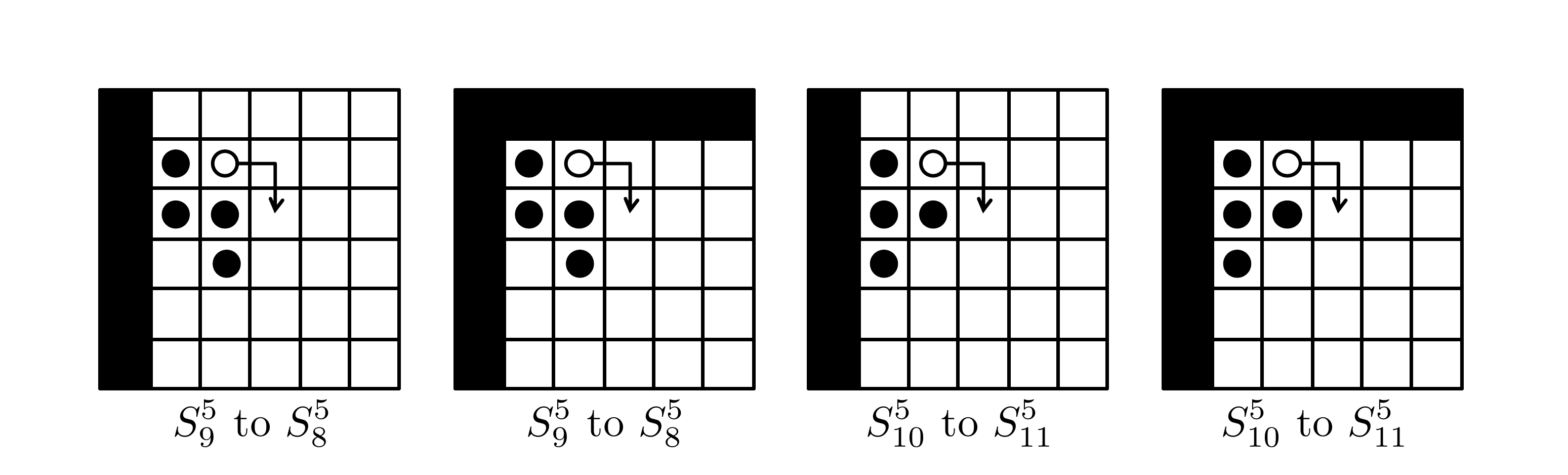}
   \caption{Exceptions with walls} 
   \label{fig:5-wall-exceptions}
  \end{minipage}

 \end{tabular}
\end{figure}

\noindent{\bf Moves along a row.~} Figure~\ref{fig:5-to-east}  
and Figure~\ref{fig:5-to-west} show the move to the east and
the move to the west, respectively.
By repeating one of the two moves, $R$ moves to one direction. 
In the beginning of the two moves,
its spine is the $i$th row an its frontier is
the $j$th column. 

\noindent{\bf Turns on the walls.~} Figure~\ref{fig:5-east-wall} 
and Figure~\ref{fig:5-west-wall} show a turn on the east wall and
a turn on the west wall, respectively. 
During the turns, each module can see the walls. 
The spine changes after a turn on the west wall, while 
it does not change after a turn on the east wall.

\noindent{\bf A turn on the southwest corner.~} When the spine of $R$
reaches the first row and it comes back to the west wall,
it turns the track by $\pi/2$ as shown in
Figure~\ref{fig:5-east-south-corner}.
Note that the cells of the $0$th row have been visited
by the modules under the spine 
when $R$ moves from the east wall to the west wall. 
The final state of the turn is $S^5_5$ and $R$ moves
along the $0$th column by the moves in Figure~\ref{fig:5-to-east}. 
Here, the spine is the first column, and the
$0$th column is visited by the modules over the spine. 

We finally add exceptional movements. 
Figure~\ref{fig:5-moving-exceptions} and \ref{fig:5-wall-exceptions}
show all states, for which no movement is defined yet. 
To be more precise, for states $S^5_1, S^5_2, \ldots, S^5_{10}$,
almost all states (including walls) are used in the proposed 
algorithm except $S^5_9$ and $S^5_{10}$ with walls. 
For states $S^5_{11}, \ldots, S^5_{18}$,
only six states with walls are used in the proposed algorithm.
Hence in the remaining states, $R$ changes its state to
one of $S^5_1, S^5_2, \ldots, S^5_{10}$ through at most two steps
as shown in Figure~\ref{fig:5-moving-exceptions} and
\ref{fig:5-wall-exceptions}. 

The reference point of $R$ visits all cells in each row and its 
progress is clear from the proposed algorithm.

By Lemma~\ref{lemma:local-nec} and \ref{lemma:local-suf}, 
we obtain Theorem~\ref{theorem:local}.

\subsection{Locomotion without global compass} 
\label{subsec:locomotion}

Theorem~\ref{theorem:local} excludes initial deadlock states 
of five modules and related initial states. 
We now discuss the minimum number of modules for self-stabilizing search. 
We start with a more fundamental problem called the \emph{locomotion problem} 
that requires the metamorphic robotic system in an infinite 2D square 
grid to keep on moving to one direction. 
Hence, the modules need to break their initial symmetry and 
agree on a moving direction. 

We formalize the symmetry of the modules in the 2D square grid. 
In a square grid, the symmetry of an initial state of the metamorphic 
robotic system is identified with three types of cyclic groups, 
$C_1$, $C_2$, and $C_4$, where 
$C_1$ consists of the identity element. 
The cyclic group $C_k$ consists of rotations by 
$2\pi/k, 4\pi/k, \ldots, 2\pi$ around a specified center. 
For example, the symmetry of Figure~\ref{fig:sym3} is $C_2$ 
because the state obtained by a rotation $\pi$ or $2\pi$ is identical 
to itself. 
The two modules at symmetric positions regarding the 
rotation center (i.e., the center cell) may have symmetric 
local coordinate systems and they may move symmetrically forever. 
The symmetry of Figure~\ref{fig:sym4} and is $C_4$, 
and the four modules cannot move because if one of them 
moves, the other three modules move simultaneously and symmetrically 
in the worst case and the backbone does not exist. 
The symmetry of Figure~\ref{fig:sym5} is also $C_4$; 
however, readers may expect that the 
module on the rotation center can break the symmetry 
by leaving the center. 
The following lemma presents a counter-intuitive fact; 
the module on the center cannot break the symmetry 
due to connectivity requirement. 

\begin{lem}
\label{lemma:loco-nec} 
When modules are not equipped with the global compass, 
the metamorphic robotic system can start locomotion from any initial state 
only if it consists of $(4m+3)$ modules for some positive integer $m$. 
\end{lem}
\begin{proof}
We present initial states of modules from which 
the metamorphic robotic system cannot start locomotion. 
We have the following three cases. 

\noindent{\bf Case 1.~}$4m$ modules. 
 Consider an initial state of $4m$ modules in which 
 they form two parallel lines, i.e., $C_0 = \{
 c_{i,j}, c_{i,j+1}, \ldots, c_{i,j+2m-1}, 
 c_{i+1,j}, c_{i+1,j+1}, \ldots, c_{i+1,j+2m-1} \}$. 
 The rotational symmetry of $C_0$ is $C_2$ with the 
 rotational center $(i+1, j+m)$. 
 In this case, every pair of modules 
 in cells $c_{i,j+k}$ and $c_{i+1, j+2m-1-k}$ 
 (or $c_{i+1, k+k}$ and $c_{i, j+2m-1-k}$)
 for $k=0,1,2,\ldots, m$ 
 perform symmetric movement in the worst case. 
 Assume that there exists an algorithm $A$ that makes 
 the $4m$ modules start locomotion from any initial state. 
 If the module in $c_{i, j+k}$ moves to 
 $c_{i+x, j+k+y}$, then the module in $c_{i,j+2m-1-k}$ moves to 
 cell $c_{i+1-x, j+2m-1-k-y}$. 
 Eventually, $x$ and $y$ becomes large enough to make 
 the modules disconnected. 
 Hence, there is no locomotion algorithm $A$ for this 
 initial state. 
 
\noindent{\bf Case 2.~}$(4m+1)$ modules. 
 Consider an initial state $C_0$ of $(4m+1)$ modules 
 where they form a large plus sign with center cell $c_{i,j}$. 
 See Figure~\ref{fig:sym5} as an example for $m=1$. 
 The rotational symmetry of $C_0$ is $C_4$ with the rotation center 
 $(i+1/2, j+1/2)$. 
 In the same discussion as the previous case, 
 the modules forming the four long lines move symmetrically in the worst case. 
 The only possibility to break their symmetry is to 
 move the module in $c_{i,j}$. 
 When the module performs a sliding, 
 one of the side-adjacent cells $c_{i+1, j}, c_{i, j+1}, c_{i-1, j}, c_{i, j-1}$ 
 must be empty. 
 Because of the symmetry, 
 if the module in one of the four cells moves, 
 the other three modules move simultaneously, 
 and the metamorphic robotic system cannot maintain its 
 connectivity. Hence, the module in $c_{i,j}$ cannot 
 perform a sliding. 
 In the same way, the module in $c_{i,j}$ cannot perform 
 a rotation because any rotation requires that one of the four 
 side-adjacent cells is empty. 
 Consequently, the module in cell $c_{i,j}$ cannot perform any 
 movement and the modules cannot perform locomotion. 
 
\noindent{\bf Case 3.~}$(4m+2)$ modules. 
 By the same discussion as the first case, 
 the modules cannot perform locomotion. 

Consequently, $(4m+3)$ modules are necessary condition 
for the metamorphic robotic system to perform locomotion 
from an arbitrary initial state. 
\qed 
\end{proof}

By Lemma~\ref{lemma:loco-nec}, 
the minimum number of modules that may perform locomotion 
from an arbitrary initial state is seven. 
Actually, there is an algorithm that transforms 
any symmetric initial state of seven modules 
to the leftmost state in Figure~\ref{figure:7-loco}, 
and the center module can break the symmetry 
by $1$-sliding to east or west.\footnote{Consequently, 
it is impossible to break the initial symmetry of $C_4$, 
while it is possible for $C_2$.}
Then, for example, the modules can perform a locomotion 
as shown in Figure~\ref{figure:7-loco}. 

\begin{figure}[t]
\centering
 \includegraphics[height=2cm]{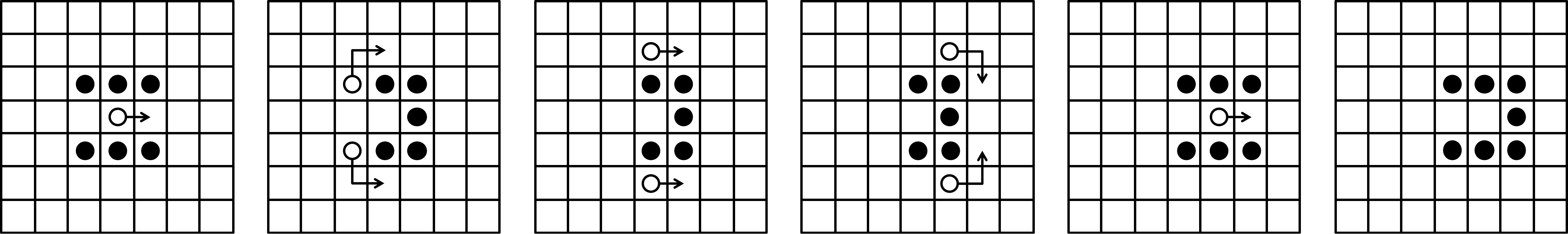}
 \caption{Locomotion of seven modules} 
 \label{figure:7-loco}
 \end{figure}

\begin{lem} 
\label{lemma:locom-suf}
Seven modules not equipped with the global compass 
are necessary and sufficient for the metamorphic robotic system 
to perform locomotion from an arbitrary initial state. 
\end{lem}

After the first movement of Figure~\ref{figure:7-loco}, 
the seven modules do not have any symmetry 
and they can agree on a total ordering 
among themselves because they agree on the clockwise direction. 
Thus, from this state, each module can perform 
different movements. 
We can easily extend the proposed search algorithm in 
Section~\ref{subsec:local-suf} 
for seven modules and we have the following result. 
The necessity is by Lemma~\ref{lemma:local-nec} and Lemma~\ref{lemma:loco-nec}.

\begin{corl}
 Seven modules not equipped with the global compass 
 are necessary and sufficient for
 the metamorphic robotic system to find a target 
 in any given field from any initial configuration.  
\end{corl}

\section{Conclusion and future work}
\label{sec:concl}

We considered three dynamic tasks for the metamorphic robotic system, 
i.e., search, exploration, and locomotion. 
We demonstrated the 
effect of the global compass on the minimum number of modules 
for search in an unknown field.  
We first demonstrated that when the modules are equipped with 
the global compass, 
three modules can find a target from an arbitrary initial sate and 
less than three modules cannot perform search. 
Then, we demonstrated that when the modules are not equipped with the 
global compass, 
five modules perform search from limited initial states 
and less than five modules cannot perform search. 
Finally, we discussed locomotion and demonstrated 
that when the modules are not equipped with the global compass, 
seven modules perform locomotion from any initial state. 
The proposed search algorithms make the metamorphic robotic system 
perform exploration of a field until it finds the target. 
When there is no target in the field, 
the proposed algorithm makes the metamorphic robotic system 
perform perpetual exploration, i.e., exploration without stop. 

The proposed search algorithms also illustrate 
the effect of global compass on the search time. 
Both search algorithms exhibit the worst search time 
when the metamorphic robotic system starts from a row 
just below the target. 
The search algorithm in Section~\ref{sec:global} makes 
the metamorphic robotic system consisting of three modules 
equipped with the global compass 
sweep all rows to find the target. 
On the other hand, 
the search algorithm in Section~\ref{sec:local} makes 
the metamorphic robotic system consisting of five modules 
not equipped with the global compass 
sweep all rows and columns. 
Intuitively, the latter algorithm imposes two visits 
of each cell in the worst case, 
while the former algorithm imposes a single visit of each cell. 
We believe that additional modules can reduce 
the second visits, that is, 
five modules are not enough for efficient exploration track. 
It remains to be investigated how many modules are necessary and sufficient 
for efficient search when the modules are not equipped with 
the global compass. 
Of course, the number of steps in a search 
can be improved with faster moves and turns. 

Another important directions is search in other fields,
for example, a convex and non-convex field in the 2D square grid, 
a finite 3D square grid, 
torus, graphs, sphere, and generalization to continuous space. 

%
%
%
%
%
%

\bibliographystyle{plain}
\bibliography{papers} 

\end{document}